%% file: main.tex
\documentclass[format=acmsmall, review=false]{acmart}

\usepackage{acm-ec-18}

\usepackage{indentfirst}
\usepackage{amsmath, amssymb, amsthm, amsfonts}
\usepackage{tikz, comment}
\usepackage{algpseudocode, algorithm}
\usepackage{tcolorbox, qtree}
\usepackage{enumerate}
\usepackage{fancyvrb}
\usepackage{caption}
\usepackage{xfrac}

\newcommand{\R}{\mathcal R}
\newcommand{\Q}{{\mathcal Q}}
\newcommand{\Z}{{\mathcal Z}}

\newcommand{\W}{{\mathbb W}}

\newcommand{\eat}[1]{}

\renewcommand {\S}{\mathcal S}
\renewcommand {\W}{\mathcal W}

\renewcommand {\P}{\mathcal P}

\newcommand {\V}{\mathcal V}

\newcommand{\F}{\mathcal{F}}

\newcommand{\C}{\mathcal{C}}

\renewcommand{\L}{\mathbf{L}}

\begin{document}

\title{Two-sided Facility Location}
\author{Reza Alijani}
\affiliation{
    \institution{Duke University}
    \department{Computer Science}
    \city{Durham}
    \state{NC}
    \postcode{}
    \country{USA}}   
    \author{Siddhartha Banerjee}
\affiliation{
    \institution{Cornell University }
    \department{ORIE}
    \city{Ithaca}
    \state{NY}
    \postcode{}
    \country{USA}}
        \author{Sreenivas Gollapudi}
\affiliation{
    \institution{Google Research}
    \department{}
    \city{}
    \state{CA}
    \postcode{}
    \country{USA}}
            \author{Kostas Kollias}
\affiliation{
    \institution{Google Research}
    \department{}
    \city{}
    \state{CA}
    \postcode{}
    \country{USA}}

\author{Kamesh Munagala}
\affiliation{
    \institution{Duke University}
    \department{Computer Science}
    \city{Durham}
    \state{NC}
    \postcode{}
    \country{USA}}



\begin{abstract}
Recent years have witnessed the rise of many successful e-commerce marketplace platforms like the Amazon marketplace, AirBnB, Uber/Lyft, and Upwork, where a central platform mediates economic transactions between buyers and sellers. Motivated by these platforms, we formulate a set of facility location problems that we term {\em Two-sided Facility location}. In our model, agents arrive at nodes in an underlying metric space, where the metric distance between any buyer and seller captures the quality of the corresponding match. The platform posts prices and wages at the nodes, and opens a set of facilities to route the agents to. The agents at any facility are assumed to be matched. The platform ensures high match quality by imposing a distance constraint between a node and the facilities it is routed to. It ensures high service availability by ensuring flow to the facility is at least a pre-specified lower bound. Subject to these constraints, the goal of the platform is to maximize the social surplus (or gains from trade) subject to weak budget balance, {\em i.e.}, profit being non-negative. 

We present an approximation algorithm for this problem that yields a $(1 + \epsilon)$ approximation to surplus for any constant $\epsilon > 0$, while relaxing the match quality ({\em i.e.}, maximum distance of any match) by a constant factor. We use an LP rounding framework that easily extends to other objectives such as maximizing volume of trade or profit.

We justify our models by considering a dynamic marketplace setting where agents arrive according to a stochastic process and have finite patience (or deadlines) for being matched. We perform queueing analysis to show that for policies that route agents to facilities and match them, ensuring a low abandonment probability of agents reduces to ensuring sufficient flow arrives at each facility.  Such an analysis also helps us posit facility location variants that capture settings where the platform elicits deadlines truthfully by posting lotteries over different prices and wages for different deadlines.
\end{abstract}

\maketitle




\input{Introduction.tex}
\input{ObliviousModel.tex}

\input{ObliviousAlgorithm.tex}

\input{justif.tex}

\input{Elicitation0.tex}

\input{ElicitationModel.tex}
\input{queueing.tex}

\section{Conclusions}
\label{sec:open}
Our work is a first step in understanding the problem of jointly pricing and scheduling in dynamic matching markets. We showed that a natural sub-class of two-sided stochastic matching policies can be reduced to novel variants of facility location, yielding approximation algorithms for the joint pricing and scheduling problem in two-sided markets.

We now mention several open questions that arise.  For the facility location problems, there are other variants that we do not yet have good algorithms for. For instance, our model imposes a uniform distance bound the match of any agent. Extending it to average match distance will require new techniques; the basic filtering step in facility location rounding fails in our case since the demand value itself is a variable. It is also an interesting question to extend our techniques to when markets can be priced, and agents choose markets to optimize their utility, extending  techniques for stochastic scheduling in one-sided markets~\cite{Chawla17a,Chawla17b}.


For the dynamic marketplace problems, we reduced a special type of scheduling policy to the facility location problems. One can ask: {\em What about approximating the overall optimal policy?}  Such a policy need not use facilities, and is poorly understood even when pricing is not involved (for instance, see~\cite{AkbarpourGharan}). 

Our  model here also assumed Poisson arrivals whose rate is constant over time. A different approach is to use online algorithms. In particular, it would be interesting to incorporate pricing and wages into the ``online matching with delays" models considered in~\cite{AzarCK17,EmekKW16}. For instance, if we ignore pricing, and assume agents simply reveal their deadlines, a natural objective is to minimize the sum of metric distances traveled by agents for their match plus the weighted number of agents whose deadlines are missed. For this problem, the techniques in~\cite{AzarCK17,EmekKW16} do not seem to extend. An interesting question is whether a polylogarithmic competitive ratio algorithm exists for this variant. 

\noindent {\bf Acknowledgment: } We thank Sungjin Im and Janardhan Kulkarni for helpful comments on an earlier draft of this paper. Reza Alijani and Kamesh Munagala are supported by NSF grants CCF-1408784, CCF-1637397, and IIS-1447554.
\bibliographystyle{ACM-Reference-Format}
\bibliography{pooling}

\newpage


\appendix
\input{appendix.tex}

\end{document}

%% file: Introduction.tex
\section{Introduction}
Online marketplaces have transformed the economic landscape of the modern world. Many of today's most important companies are platforms facilitating trade between agents: both for goods (Amazon, eBay, Etsy), and increasingly, services: transportation (Lyft, Uber); physical and virtual work (Taskrabbit, Upwork); lodging (Airbnb); 
etc. These platforms enable fine-grained monitoring of participants, and greater control via pricing
, terms of trade, recommendation and directed search
, etc. The challenge of harnessing this increase in data and control has led to a growing literature in online marketplace design.

The basic algorithmic challenge facing a marketplace platform can be summarized as follows: it must decide \emph{which buyer should match to which seller, at what time, and for what price and wage}, in order to maximize some desired objective. In a sense, this combines the challenges of online bipartite matching, job scheduling, and pricing and mechanism design.   
In this paper, we build towards a model that simultaneously addresses these three orthogonal challenges and tackles the two-sided marketplace design problem in its full generality. 

\subsection{High-Level Problem Formulation}
Our motivating application is the design of a marketplace platform, and in particular, policies for pricing both buyers and sellers, and scheduling feasible matches between them.  In our setting, buyers and sellers flow into nodes located in a metric space representing an underlying feature space in which agents are embedded. The metric distance between a buyer and seller captures the overall quality of the match; For instance, in AirBnB, it represents some combination of geographic distance from the desired location and quality of room; in UpWork, it can capture how well the skill set of the consultant matches the task requirements; and so on. We assume the platform needs to match buyers and sellers within a distance threshold specified as input.

In addition to their known embedding in the metric space, buyers have value for being served, and  sellers have costs for providing service; the platform knows the distribution of these values.  The platform needs to make two policy decisions: what prices/wages to set at each node, and how to match the resulting flow of buyers and sellers.  We consider a simple class of policies where the platform opens a set of facilities  in the metric space (or {\em canonical features} in the feature space) and routes flow from demand/supply nodes to these facilities. The flow routed to any facility is matched up. In effect, the platform maps agents' features to some close-by canonical features and matches up agents with the same features. We insist these facilities satisfy the following service guarantees:
\begin{enumerate}
\item \textbf{Quality of service guarantees}: The flow assigned to a facility is from supply and demand nodes within distance $R$. Therefore, any matched demand/supply is within distance $2R$.
\item \textbf{Service availability guarantees}: In order for the flow at the facility to match up, each facility needs to have {\em flow balance} of supply and demand routed there. Furthermore, if we view facilities as {\em canonical features/types} that are used to compress the feature space of buyers/sellers, then we want these features to be sufficiently representative. We therefore insist these facilities are sufficiently {\em thick}, meaning there is a {\em lower bound} $\L$ on the flow routed  to each of them. 
\end{enumerate}

In Section~\ref{sec:justif}, we present a more formal queueing-theoretic justification, showing that the above constraints arise from natural stochastic matching policies for dynamically arriving {\em atomic} agents.

The goal of the platform is to (i) set prices/wages at each node in the metric space; (ii) open a set of facilities; and (iii) route the resulting demand/supply flow to these facilities satisfying the quality and service availability guarantees. The objective of the platform is to maximize the {\em gains from trade} or {\em social surplus}, which is total value of buyers {\em minus} total cost of sellers, subject to {\em weak budget balance}, meaning that the  profit (the difference between total price charged to buyers and total wage paid to sellers) is non-negative. Our same solution idea applies to other objectives, such as maximizing profit, or maximizing the {\em throughput} or volume of trade subject to weak budget balance. This leads to a class of optimization problems that we term {\em Two-sided Facility Location}. 

\subsection{Our Results} 
 Our first main contribution in Sections~\ref{sec:oblivious}--\ref{sec:rounding} is to show an approximation algorithm for {\em Two Sided Facility Location}. We present a new LP rounding framework that for any constant $\epsilon > 0$, achieves a $(1+\epsilon)$ approximation to the social surplus objective (resp. throughput and profit objectives). It relaxes the distance bound constraint by a factor of $4$, while preserving the budget balance constraint, as well as the flow balance and lower bound constraints at each facility.  If we allow a tiny additive error $\Delta$ in the surplus objective, our algorithm requires solving $O\left( \frac{n^{1/\epsilon}}{\epsilon} \log \frac{n W_{\max}}{\epsilon \Delta} \right)$ LPs, where $n$ is the number of nodes, and $W_{\max}$ is the maximum possible surplus.
 

We show in Section~\ref{sec:hardness} that the surplus objective is {\sc NP-Hard} to approximate to a factor $o(\L^c)$ for some constant $c > 0$, unless the distance bound is relaxed by at least a factor of $2$. 

\medskip
\noindent{\bf Techniques.} Our facility location variants mirror the {\em profit earning facility location} problem in~\cite{PEFL}. Just like that setting, we have lower bounds on demand served at each facility and an upper bound on how far the facility can be from an assigned demand. However, there are key differences that preclude the application of existing techniques from lower-balanced facility location~\cite{ZS16,GuhaMM00,KargerM00,PEFL,zoya}: First, the demand or supply at each node is a variable that can be adjusted using pricing. This means the demand/supply can be zero at some ``outlier" nodes, so that they do not need to be served by any facility. Secondly, each facility needs to satisfy flow balance between supply and demand, and finally, both surplus and profit involve {\em differences}, so the platform can potentially lose money at some facilities, but recover it at others. 

The above differences make formulating an LP relaxation tricky. Note that even in~\cite{PEFL}, the version with outliers and profits that can become negative has unbounded integrality gap, because the optimal profit can be zero while the LP achieves positive profit. Unlike~\cite{PEFL}, since we can control demand/supply by pricing, we have greater flexibility in modifying the LP variables.  Despite this, the integrality gap of the straightforward LP formulation for our problem is large, because there could be a facility that generates a bulk of the surplus, but has large negative profit that is compensated by other facilities. (See Appendix~\ref{app:example} for an example). 

This brings up our main technical contributions: We first observe that if we focus on the LP variables corresponding to facility $i$, we can  scale these up or down by changing the fraction to which this node is an outlier. This enables us to use techniques reminiscent of improved greedy algorithms for budgeted coverage problems~\cite{KhullerMN,DeanGV}: In particular, we {\em strengthen the LP formulation} via guessing a few of the facilities that are opened in the optimal solution.  Next, we use the guesses to develop a structural characterization for this stronger LP based on modifying variables for {\em pairs} of facilities.  In effect this shows that there is some integrality in the neighborhood of any partially open facility, which helps us consolidate these facilities while preserving all constraints. 

\medskip
\noindent {\bf Queueing-theoretic Foundations.} Our second main contribution is to present {\em queueing-theoretic} micro-foundations to our facility location problems. What makes the problems hard and technically interesting is the presence of the lower bound $\L$ on flow routed to each facility; if $\L = 0$, the problem is in fact polynomial time solvable!  In Section~\ref{sec:justif}, we show that the lower bound constraints (along with flow balance) arise naturally from a {\em dynamic marketplace} setting. Here, atomic buyers/sellers arrive according to Poisson processes, and have a private patience for being matched before they abandon the system. In this context, service availability corresponds to joint pricing and matching policies that have low abandonment rate. We perform queueing analysis to show that for policies that route agents to close-by facilities and match them optimally there, ensuring a low abandonment probability  reduces to ensuring  both flow balance and that sufficient flow arrives at each facility. In effect, our {\em static} facility location variants have their roots in {\em dynamic} control policies that maximize surplus while ensuring low abandonment. 

The advantage of our LP rounding framework for {\em Two-sided Facility Location} is that it gracefully handles more complex variants motivated by the dynamic marketplace setting. For instance, consider the dynamic problem where the platform uses prices and wages to truthfully elicit patience of agents, and subsequently matches them optimally using Earliest Deadline First (EDF) scheduling in each facility. Motivated by this, we consider an {\em envy-free} variant of {\em Two-sided Facility Location}  in Section~\ref{sec:elicited}, where each node (or agent type) is composed of sub-types that envy one another. The platform sets prices/wages for each sub-type so that each agent truthfully chooses its sub-type. Each sub-type has a weight and the service availability constraints are captured by flow balance of supply and demand at each facility, and a lower bound on the total {\em weight} routed to each open facility. Our LP rounding framework easily extends  to yield optimal profit while relaxing the distance constraint by factor $4$.

 
 \subsection{Related Work}
\label{sec:related}

\noindent {\bf Two-sided Markets.} Our objective  maximizes social surplus subject to budget balance (and individual rationality). This is a classic objective in two-sided market mechanisms, and originates in the celebrated work of Myerson and Satterthwaite~\cite{MyersonS}, where it is termed {\em gains of trade}. They considered the case of a single buyer and seller.  
This has inspired a recent line of work on truthful mechanisms for approximate surplus maximization in markets of multiple buyers and sellers~\cite{mcafee2008gains,BrustleCWZ17,Blumrosen}, ultimately resulting in a $2$-approximation to gains of trade. This line of work assumes buyers and sellers are matched in one shot. The novelty in our work is in modeling a {\em dynamic setting} and incorporating {\em service availability} guarantees while preserving the same objectives. We therefore consider the more natural class of mechanisms that post prices and wages. Posted price mechanisms have been extensively studied  in two-sided marketplaces~\cite{mcafee2008gains,rochet2003,armstrong2006,weyl2010}, and the main idea we borrow from this literature is the notion of {\em insulating tariffs}~\cite{weyl2010}, which posits that market design is easier if the prices seen by buyers is disconnected from the wages seen by service providers.  

Another recent line of work shows approximately optimal mechanisms for maximizing welfare in two sided markets with {\em goods}~\cite{Leonardi1,Leonardi2}; however, theirs is a sum objective defined in terms of the final sets of items allocated to each buyer and seller, which is different from the gains of trade. 



Finally, our work is related to two-sided segmentation problems considered in~\cite{Sid-WWW17}. In their model, prices appear endogenously via market clearing (instead of being set by the platform). Our concept of clustering buyer and seller nodes via facilities is a form of market segmentation. However, unlike~\cite{Sid-WWW17}, the flow lower bound prevents us from arbitrarily splitting a market into smaller sub-markets, making our problem technically different.

\medskip
\noindent {\bf Dynamic Marketplaces.} Our work on dynamic marketplaces is related to several recent works on online scheduling under stochastic arrivals of tasks on machines with limited resources~\cite{Hajiaghayi05,Chawla17a,Chawla17b}. Tasks have (private) types comprising their value, arrival time, and deadline; the platform's goal is to maximize welfare while truthfully eliciting the type. While similar to our work on pricing resources or tasks, they allow agents to choose assignments based on posted prices ({\em envy-freeness}). Another difference is the markets considered in their studies are one-sided. 
Blum {\em et al.}~\cite{Blum02} consider online two-sided markets with {\em fixed bids}, and present competitive algorithms for maximizing profit or number of matches. Though their objectives are simpler, their methods will necessarily discard a constant fraction of feasible matches, which can lead to significant user dissatisfaction. In contrast, our focus is on  ensuring high service availability which corresponds to small user abandonment.  

 Dynamic two-sided markets also serve as motivation for recent work on ``online matching with delays"~\cite{EmekKW16,AzarCK17}. Here, buyers and sellers arrive online in a metric space, and can be matched at any time subsequently. The goal is to minimize the total distance cost plus waiting cost, and the authors present a log-competitive algorithm. These models do not incorporate pricing. Further, our dynamic marketplace models are more closely related to dynamic matchings with {\em stochastic} arrivals, and we review this literature in Section~\ref{sec:justif}.

%% file: ObliviousModel.tex
\section{Problem Statement}
\label{sec:oblivious}
\label{sec:welfare}
There is a metric space $G(V,E)$ with an associated distance function $c$. We assume buyers and sellers are fluid and arrive at nodes in this metric space and must be matched with each other. The metric distance captures the match quality -- if  demand at node $j$ is matched to supply at node $j'$, the disutility to the system is captured by $c(j,j')$. 


\medskip \noindent {\bf Demand and Supply Functions.} Each node $j \in V$ is associated with a demand function $F_j$ and a supply function $H_j$. When offered price $p$, we assume the demand ({\em i.e.}, buyers)  at node $j$ 
is $d_j F_j(p)$, where $F_j(p)$ is a non-increasing function of $p$ corresponding to the survival function of a continuous density function $f_j$ on valuations; formally $F_j(p) = \int_{v=p}^{\infty} f_j(v) dv$.
In other words, the volume of buyers is $d_j$, and when quoted a price $p$, only buyers with valuations at least $p$ choose to participate, and hence the stream of buyers is thinned by a factor $F_j(p)$. We assume there is a finite price $p_{\max}$ so that $F_j(p_{\max}) = 0$ for all $j \in V$.

Similarly, when offered wage $w$, the supply of sellers arriving at node $j$ is  $s_j H_j(w)$, where $H_j(w)$ is a non-decreasing function of $w$, corresponding to the CDF of a continuous density function $h_j$ on costs; formally $ H_j(w) = \int_{c=0}^{w} h_j(c) dc$. When offered wage $w$, all sellers with cost at most $w$ participate, resulting in supply $s_j H_j(w)$.  We assume that $H_j(0) = 0$, {\em i.e.}, sellers accrue $0$ utility by not participating in the platform. 

In this section, we make the standard regularity assumptions ({\em \`{a} la} Myerson-Satterthwaite~\cite{MyersonS}) on the density functions $f_j$ and $h_j$. In particular, we assume $x F_j^{-1}(x)$ is concave in $x$ and $y H_j^{-1}(y)$ is convex in $y$. This is true for instance, for all log-concave densities $f_j$ and $h_j$, which includes Normal, Exponential, and Uniform distributions. In Section~\ref{sec:elicitation}, we discuss how the continuity and regularity assumptions can be removed by using lottery pricing.



\medskip\noindent {\bf Facilities.} As discussed before, 
 the platform opens a set of facilities or ``canonical features" in the metric space and routes demand and supply fractionally to open facilities within distance bound $R$ in order to match them up.  (Therefore, the distance of any match is at most $2R$.) 
It ensures high service availability by ensuring the following two intuitive constraints on the flow of supply and demand arriving each open facility. We justify these constraints via queueing arguments in Section~\ref{sec:justif}.

\begin{description}
\item[Flow Balance.] Since the flow arriving at the facility are matched, the total amount of supply and demand are equal.
\item[Flow Lower Bound.] The facility is sufficiently {\em thick}, that is, the total amount of supply (resp. demand) is at least $\L$.
\end{description}



\newcommand{\TSFL}{{\sc Two-Sided Fac-Loc}}

\subsection{Two-sided Facility Location}
We now have all the ingredients to make our problem, \TSFL$(\L,R)$, precise. 
Let $\F \subseteq V$  the set of all candidate facilities; we set $\F = V$.  For each node $j$, $B_R(j) \subseteq \F$ denotes the set of all the  facilities $i \in \F$ such that $c(i,j) \le R$. Similarly, for each facility $i$, we define $B_R(i)$ as the set of all the nodes $j \in V$ such that $c(i,j) \le R$.  A solution to \TSFL$(\L,R)$ is specified by the following: 
\begin{itemize}
\item An assignment of price $p_j$ and wage $w_j$ to each node $j \in V$. If the price (resp. wage) at node $j$ is $p_{\max}$ (resp. $0$), we assume this node generates no demand (resp. supply).
\item A set of locations $S \subseteq \F$ for opening the facilities; and
\item A routing scheme $\vec{x^d_j}$ (resp. $\vec{x^s_j}$) for each demand (resp. supply) node $j \in V$ that generates non-zero demand (resp. supply). For $i \in S$, if $x^d_{ij} > 0$ then $i \in B_R(j)$. Further, $\sum_{i \in B_R(j)} x^d_{ij} = 1$ for all nodes $j \in V$ that generate non-zero demand; similarly,  $\sum_{i \in B_R(j)} x^s_{ij} = 1$ for each $j \in V$ with non-zero supply.  
\end{itemize}

Note that the flow of demand (resp. supply) from node $j \in V$ to facility $i \in S$ is $d_j F_j(p_j) x^d_{ij}$ (resp. $s_j H_j(w_j) x^s_{ij}$).  We enforce that the flows satisfy the {\em flow balance} and {\em flow lower bound} conditions at each $i \in S$. 

\medskip
\noindent {\bf Weak Budget Balance.} The next constraint is {\em weak budget balance}, which corresponds to profit being non-negative. This is written as:
$$ \mbox{Profit} = \sum_{j \in V}  \left(  d_j p_j F_{j}(p_j)   -  s_j w_j H_{j}(w_j)  \right) \ge 0$$

\medskip
\noindent {\bf Surplus (Gains from Trade) Objective.} We first define the following quantities:
$$
\V_j(p) = d_j \int_{v=p}^{\infty} v f_j(v) dv \qquad \mbox{and} \qquad \C_j(w) = s_j \int_{c=0}^{w} c h_j(c) dc
$$
respectively denote the total value of buyers generated by node $j$ when the price there is $p$ and the cost of sellers at node $j$ when the wage there is $w$. The surplus objective can then be written as:
$$ \mbox{Social Surplus = }  \sum_{j \in V}  \left(  \V_{j}(p_j)   -  \C_{j}(w_j)  \right)$$
This defines the problem \TSFL$(\L,R)$

\medskip
Though we focus on surplus in the paper, the same techniques extend to other objectives such as maximizing {\em throughput} or volume of matches, defined as the total demand (or supply): $\sum_{j \in V} d_j F_j(p_j) $. It also extends to the objective of maximizing profit defined above.

%% file: ObliviousAlgorithm.tex
\newcommand{\E}{\mathbf{E}}

\subsection{Hardness of Approximation}
\label{sec:hardness}
We characterize the approximation ratio of any algorithm for \TSFL$(\L,R)$ as $(\alpha,\gamma)$, if the resulting solution relaxes the distance bound of an assignment to a facility to $\alpha R$, ensures lower bound $\L$, and has surplus $OPT/\gamma$, where $OPT$ is the optimal surplus.  

\begin{theorem}
It is {\sc NP-Hard} to find a $(\alpha, \gamma)$ approximation for \TSFL$(\L,R)$ unless $\alpha \ge 2$ or $\gamma \ge \L^{c}$ for some constant $c > 0$. 
\end{theorem}
\begin{proof}
We reduce from Maximum Independent Set in $k$-regular graphs ($k$-MIS). Given a $k$-MIS instance with $n$ vertices and $m = kn/2$ edges, construct a metric space where each edge in the $k$-regular graph $G(V,E)$ has length $2R$. Place a demand node at the mid-point of each edge, and a supply node at each vertex. We set $\L = k$. Each supply node has $s_j = k$, and supply function $H^{-1}(r) = 1 - \delta$ for $r \in [0,1]$.  Similarly, each demand node has $d_j = 1$, and demand function $F^{-1}(q) = 1$ for $q \in [0,1]$. Since the distance threshold is $R$, the facilities are opened at vertices of the graph. Each such facility must see $k$ units of supply and demand, which means all neighboring demand is routed there, leading to surplus (resp. profit) $k \delta$ at that facility. Since two open facilities cannot share a demand, this means the open facilities form an independent set. Therefore, the surplus of \TSFL$(k,R)$ is $\delta$ times the size of the maximum independent set in $G$. This is {\sc NP-hard} to approximate to within a factor of $k^{c}$ for some constant $c > 0$; see~\cite{Alon,Elad}. Therefore, we need to relax the distance bound by  factor of $2$.
\end{proof}

 In the sequel, we present a $(4,1+\epsilon)$ approximation.   For the algorithm to have polynomial running time, they also need lose a small additive amount in the objective; as we show later, this quantity can be exponentially small.

\section{Linear Programming Relaxation}
\label{sec:basic}

We now formulate  \TSFL$(\L,R)$ as an integer linear program.  
For ease of exposition, we compare against an optimal solution that is restricted to using prices from a fixed set $\P$ and wages from a fixed set $\W$.  Our solution is not restricted to using prices and wages from this set. In Appendix~\ref{sec:compute}, we show that our LP admits to a polynomial time solution of arbitrary additive accuracy when this assumption is relaxed, and demand/supply distributions are continuous.

Note that we assume $p_{\max} \in \P$ and $0 \in \W$, and at this price (resp. wage) the demand (resp. supply) is identically zero. This is the price (resp. wage) where this node becomes an outlier and the solution is not required to open a facility nearby.

Instead of writing our LP using prices and wages, we use the associated demand/supply values. Let  $ \Q_j = \{q \ | q =  F_j(p),  p \in \P \}$ and  $\R_j = \{r \ | \ r = H_j(w),  w \in \W \} $. The case where the node is an outlier now corresponds to setting $q = 0$ (resp. $r = 0$). 

We redefine the valuations and costs using supply/demand values as follows:
\begin{equation}
\label{eq:valuation} 
\V_j(q) = d_j \int_{v=F_j^{-1}(q)}^{\infty} v f_j(v) dv \qquad \mbox{and} \qquad \C_j(r) = s_j \int_{c=0}^{H_j^{-1}(r)} c h_j(c) dc
\end{equation}
respectively denote the total value of buyers generated by node $j$ when the price there is $F_j^{-1}(q)$, and  the cost of sellers at node $j$ when the wage there is $H_j^{-1}(r)$. 

\medskip
\noindent{\bf Variables.}  For each candidate facility $i \in \F$, let $y_i \in \{0,1\}$ be the indicator variable that a facility is opened at that location in the metric space. Let $\alpha_{jq} = 1$ if the price at node $j \in V$ corresponds to  $q \in \Q_j$. Similarly define $\beta_{jr}$ for $r \in \R_j$.  The variable $z_{ijq}$ is non-zero only if $\alpha_{jq} = 1$ and $i \in B_R(j)$. In this case, it is the fraction of $j$'s demand that is routed to $i$. We define $z_{ijr}$ similarly for supply. Note that the actual flow from $j$ to $i$ is  $d_j q z_{ijq}$; similarly for sellers. 

\medskip
\noindent{\bf Objective and Weak Budget Balance.} The objective of social surplus  and the profit being non-negative can be captured by:
\begin{eqnarray}
\label{eq:rev0}
\label{eq:WBB}
\mbox{Surplus Objective:} & \max \sum_{j \in V} \left( \sum_{q \in \Q_j} \alpha_{jq}  \V_j(q) - \sum_{r \in \R_j} \beta_{jr} \C_j(r)\right) &\\
 \mbox{Weak Budget Balance:} &  \sum_{j \in V}  \left( \sum_{q \in \Q_j} \alpha_{jq}  d_j q F^{-1}_{j}(q)   - \sum _{r \in \R_j} \beta_{jr} s_j r H^{-1}_{j}(r)  \right) &\ge  0
\end{eqnarray} 

\noindent{\bf Feasibility.} The following constraints connect the variables together. We present these constraints only for buyers (that is, $q \in \Q_j$); the constraints for sellers is obtained by replacing $q$  with $r  \in \R_j$. First, for each $q \in \Q_j$, we need to choose one price for buyers (resp. sellers). 
\begin{eqnarray}
\label{eq:single}
\sum_{q \in \Q_j} \alpha_{jq}  & = &  1 \quad \forall j \in V \\
\label{eq:singleprime}
\sum_{i \in B_R(j)}  z_{ijq}  & = &  \alpha_{jq} \quad \forall j \in V, q \in \Q_j 
\end{eqnarray}

Next, if demand is fractionally routed from $j$ to $i$, then $i$ should be open and within distance $R$. Note that we need to ignore the case where $q = 0$ (resp. $r = 0$) since in this case, the demand (resp. supply) routed is zero, so that there is no need for a nearby facility.
\begin{eqnarray}
\label{eq:open0}
\sum_{q \in \Q_j, q > 0} z_{ijq}  & \le &  y_i\quad \forall j \in V, i\in B_R(j)
\end{eqnarray} 

\noindent {\bf Service Availability.} We finally encode {\em flow balance} and {\em flow lower bound} at each facility: 
\begin{eqnarray}
\label{eq:flowbalance0}
\sum_{j \in B_R(i)} d_j  \sum_{q \in \Q_j} q z_{ijq} & = &  \sum_{j \in B_R(i)} s_j \sum_{ r \in \R_j}  r z_{ijr}  \quad \forall i \in \F \\
\label{eq:lowerbound0}
\sum_{j \in B_R(i)} d_j   \sum_{q \in \Q_j}  q z_{ijq}  & \ge & \L y_i \quad \qquad \qquad \qquad \forall i \in \F
\end{eqnarray}

If we replace the integrality constraints on $\{y_i \}$ and the $\{\alpha_{jq}, \beta_{jr} \}$ with $y_i, \alpha_{jq}, \beta_{jr} \in [0,1]$, the above is a linear programming relaxation of the problem.   

\subsection{Integrality Gap and Stronger LP Formulation}
The main technical hurdle arises because of the flow lower bound constraint: The LP optimum (and even an integer optimum) can now open facilities $i$ which have positive surplus but negative profit, and compensate for the loss in profit by other facilitiess with positive profit.  Note that Constraints (\ref{eq:open0}) and~(\ref{eq:single}) together imply:
$$\sum_{i \in B_R(j)} y_i  \ge   1 - \alpha_{j0} \ \ \  \forall \mbox{ Demand nodes } j $$
and similarly for supply nodes. We call the quantities $\alpha_{j0}$ (resp. $\beta_{j0}$) the {\em outlier fraction} of node $j$, and correspond to the case where the node is priced in such a way that it does not generate flow.  In this case, there is no need to open a facility to satisfy $j$.  Therefore, if $\alpha_{j0} \beta_{j0} > 0$, then the above constraints could imply $\sum_{i \in B_R(j)} y_i < 1$. This means there could only be a small fractional facility open in the vicinity of $j$, which can account for a lot of the surplus. This makes the LP have super-polynomial integrality gap and we present an example in Appendix~\ref{app:example}.

\medskip
We therefore add constraints to the above LP formulation so that has bounded integrality gap.  Before showing how to strengthen the LP, we present the following easy claim, which implies that once we round $\{y_i\}$, the remaining solution can easily be made integral.

\begin{lemma}
\label{lem:increase}
Given any feasible LP solution, there is an equivalent solution that assigns only one price (resp. wage) per demand (resp. supply) node, that preserves all constraints and does  not decrease the objective.
\end{lemma}
\begin{proof}
The rounding of $\alpha_{jq},\beta_{jr}$ is simple. Let  $ \hat{q}_j = \sum_{q \in \Q_j}  q \alpha_{jq}$ and $\hat{r}_j = \sum_{r \in \R_j}  r \beta_{jr}$. Set the price of location $j$ to be $F_j^{-1}(\hat{q}_j)$ and the wage at $j$ to be $H_j^{-1}(\hat{r}_j)$. In other words, set $ \hat{\alpha}_{j\hat{q}_j} \leftarrow 1$ and $\hat{\beta}_{j\hat{r}_j} \leftarrow 1$. Further set $ \hat{z}_{ij\hat{q}_j} \leftarrow \sum_{q \in \Q'_j} z_{ijq} \frac{q}{\hat{q}_j}$ and $\hat{z}_{ij\hat{r}_j} \leftarrow \sum_{r \in \R'_j} z_{ijr} \frac{r}{\hat{r}_j}$.

Note that this process preserves the demand and supply from node $j$ to  facility $i$, which preserves all the constraints in the LP formulation. Note further that the function $q F_j^{-1}(q)$ is concave in $q$ by the regularity of the demand function. Therefore,
$$ \sum_q \alpha_{jq} q F_j^{-1}(q) \le \left(\sum_q \alpha_{jq} q\right) F_j^{-1}\left(\sum_q \alpha_{jq} q\right) = \hat{q}_j F_j^{-1}(\hat{q}_j)$$
Similarly, since we assumed $r H_j^{-1}(r)$ is convex in $r$ (by regularity of supply), we have $ \sum_r \beta_{jr} r H_j^{-1}(r) \ge  \hat{r}_j H_j^{-1}(\hat{r}_j)$. Therefore,  this transformation preserves weak budget balance. Next, we note that $\V_j(q)$ is always a concave function of $q$ and $\C_j(r)$ is always a convex function. Therefore, the above argument also implies the social surplus does not decrease in the above transformation.
\end{proof}

Define a variable for the surplus $W_i$ and profit $R_i$ of facility $i$ respectively as:
\begin{eqnarray}
\label{eq:welfare}
W_i & = & \sum_{j \in B_R(i)}  \left( \sum_{q \in \Q'_j} \V_j(q) z_{ijq} -  \sum_{r \in \R'_j} \C_j(r) z_{ijr} \right) \\
\label{eq:revone} 
R_i & = & \sum_{j \in B_R(i)}  \left( \sum_{q \in \Q_j'} d_j q F^{-1}_{j}(q)  z_{ijq}  - \sum _{r \in \R_j'} s_j r H^{-1}_{j}(r) z_{ijr}   \right)
\end{eqnarray}
Then the objective can be rewritten as: Maximize $\sum_i W_i$, and weak budget balance is $\sum_i R_i \ge 0$. Further note that $W_i \ge R_i$ since  for any $q,r$, we have $\V_j(q) \ge d_j q F_j^{-1}(q)$, and $\C_j(r) \le s_j r H_j^{-1}(r)$ if we integrate the expressions in Equation~(\ref{eq:valuation})  by parts.

\medskip
\noindent{\bf Stronger LP Formulation.} Let $\epsilon > 0$ be any constant, and let $\theta = \frac{1}{\epsilon}$. We guess the $\theta$ facilities in the optimum solution that have the most surplus. There are two cases. First, if the optimum solution opens fewer than $\theta$ facilities, we can perform a brute force search over all  integer solutions that open at most $\theta$ facilities. This can be done in $O(n^{\theta})$ time, where $n = |V|$. For each selection of facilities, Lemma~\ref{lem:increase} implies that solving the LP formulation with the corresponding $y_i$ set to $1$ and the rest to zero yields the optimal surplus (or results in declaring infeasibility). We can therefore find the surplus maximizing solution among these in polynomial time; call this surplus $W_1$.

In the other case, the optimum solution opens more than $\theta$ facilities. In this case, for every choice of parameter $\W \ge 0$ scaled to powers of $(1+\epsilon)$, and every subset $S \subseteq \F$ with $|S| = \theta$, define $LP(\W,S)$ as having all of Constraints (\ref{eq:WBB}) -- (\ref{eq:lowerbound0}), {\em plus} the following new ones:

\begin{eqnarray}
\label{eq:tighter1}
W_i & \le & \W y_i  \qquad \forall i \in \F \setminus S \\
\label{eq:tighter2}
 y_i  & = &  1 \qquad  \qquad \forall i \in S \\
 \label{eq:tighter3}
 \sum_{i \in S} W_{i} & \ge & \W \theta (1-\epsilon) 
 \end{eqnarray}

Let $OPT$ denote the optimum surplus, and let $W_2 = \max\{LP(\W,S) \ | \W \ge 0,\ S \mbox{ s.t. } |S| = \theta\} $. Then, it is easy to see that $ OPT \le \max \left(W_1, W_2 \right)$: If $OPT$ opens fewer than $\theta$ facilities, then clearly $W_1 \ge OPT$, since $W_1$ opens all possible choices of at most $\theta$ facilities.  Otherwise, let $W^*_i$ denote the surplus generated by open facility $i$ in $OPT$. Let $W^*$ denote the $\theta^{th}$ largest value of $W^*_i$. Choose  $\W \in [W^*,W^*(1+\epsilon)]$, and $S$ as the set of $\theta$ facilities in $OPT$ with $W^*_i \ge W^*$.  This induces a feasible solution to the above constraints, so that the LP optimum is at least $OPT$.


\subsection{Structural Characterization of LP Optimum}
\label{app:struct}
We now present a structural characterization about the LP optimum. This  is crucial for the rounding that we present subsequently, since it allows sufficient mass of facility to be located in roughly the same neighborhood.  

Recall $\alpha_{j0}, \beta_{j0}$ are the fractions to which node $j$ is an outlier, {\em i.e.} has zero flow. These variables are the reason the simpler LP had large integrality gap, since they allow facilities in $B_R(j)$ to be open to small fractions. Our main observation is the following: 
\begin{lemma} [Structural Characterization]
\label{lem:struct0} 
There is a $(1+\epsilon)$ approximation to the objective of  $LP(\W,S)$ that satisfies:
$$ \forall i \in \F, \qquad y_i \in (0,1) \qquad \Rightarrow \qquad \exists j \in B_R(i) \mbox{  s.t.  } \alpha_{j0} \beta_{j0} = 0$$
\end{lemma}

\medskip
\noindent {\bf High level Idea.} Before presenting the proof, we present the high level idea. Consider a facility that violates the statement. If it has $R_i > 0$, then consider all LP variables $\{y_i, z_{ijq}, z_{ijr}\}$ corresponding to some such facility $i$ and uniformly increase them. This increases both profit and surplus. We can decrease the fractions $\{\alpha_{j0}, \beta_{j0} \}$ to which any node $j$ connected to $i$ is assigned as outlier to compensate the fraction to which it is assigned $i$.  Note that Constraints~(\ref{eq:flowbalance0}) and (\ref{eq:lowerbound0}) are local to a single facility. Since we scale up {\em all} variables corresponding to a facility, we preserve these constraints. If we keep up this process, then either the facility  is completely open ($y_i = 1$); or some demand/supply node assigned to it has $\alpha_{j0} = 0$ or $\beta_{j0} = 0$. (This must hold in the LP optimum.)

On the other hand, if $W_i > 0$ but $R_i < 0$,  then increasing  its LP variables would hurt profit, which may violate the budget balance constraint; while reducing the variables would increase profit but hurt the surplus. The idea now is the following: Take any pair of such facilities; increase the variables for one facility while decrease them for the other. There is always a way of doing this so that {\em both} the total profit and surplus do not decrease -- this is essentially a fractional knapsack argument. Again, since we uniformly scale all variables corresponding to a facility, we preserve all constraints. Note that the process can also stop when a facility closes ($y_i = 0$). 
Eventually, we run out of pairs, so that for all but one facility, the above characterization holds.

At this point, the strengthened LP kicks in. The singleton facility violating the above lemma was fractionally open and had $R_i < 0$. It has surplus at most $\W$ by Constraint~(\ref{eq:tighter1}). But we have  integrally open facilities that generate surplus at least $\W\frac{(1-\epsilon)}{\epsilon}$ by Constraint~(\ref{eq:tighter3}), which means closing the singleton facility reduces surplus by at most $(1-\epsilon)$, and preserves budget balance.

\medskip 
\noindent {\bf Proof of Lemma~\ref{lem:struct0}.}
We first simplify the LP. Let $\Q_j' = \Q_j \setminus \{0\}$ and $\R_j' = \R_j \setminus \{0\}$. Let  $\eta_j = \sum_{q \in \Q_j'} \alpha_{jq}$ and  $\phi_j = \sum_{r \in \R_j'} \beta_{jr}$ respectively denote the fractions to which $j$ is assigned prices (wages) that correspond to non-zero demand (supply). We can rewrite the constraints (\ref{eq:single}) and (\ref{eq:singleprime}) as:
\begin{equation}
\label{eq:single3} 
\eta_j   =   \sum_{q \in \Q_j'} \sum_{i \in B_R(j)}  z_{ijq}  \le   1 \qquad \mbox{and} \qquad \phi_j  =   \sum_{r \in \R_j'} \sum_{i \in B_R(j)}  z_{ijr}   \le  1 \qquad  \forall j \in V 
\end{equation}
and set $\alpha_{j0} = 1 - \eta_j$, and $\beta_{j0} = 1 - \phi_j$.  Recall from Equations~(\ref{eq:welfare}) and~(\ref{eq:revone}) that $W_i$ and $R_i$ are respectively the surplus and profit of facility $i$ in the LP optimum. 

\medskip
We call node $j$ {\em fully demand-utilized} if $\eta_j = 1$, and {\em fully supply-utilized} if  $\phi_j = 1$. We say that node $j$ is {\em partially demand-connected} to facility $i \in \F$ if $\sum_{q \in \Q_j'} z_{ijq} > 0$, and {\em partially supply-connected} if $\sum_{r \in \R_j'} z_{ijr} > 0$. Let $J_D(i)$ denote the set of nodes that are partially demand-connected to $i \in \F$, and $J_S(i)$ be the set that is partially supply-connected.

\begin{lemma}
\label{lem:struct2}
\label{lem:struct}
In the LP optimum, for any $i \in \F$ with $y_i > 0$, we have $W_i > 0$. Furthermore, all except one facility with $y_i > 0$ satisfy the following condition:  either $y_i = 1$; or there exists $j \in J_D(i)$, such that $j$ is {\em fully demand-utilized}; or there exists $j \in J_S(i)$ such that $j$ is {\em fully supply-utilized}.
\end{lemma}
\begin{proof}
First, note that $W_i \ge R_i$. 
Suppose an open facility has $W_i \le 0$. This implies $R_i \le 0$. Consider a different solution that sets $y_i = 0$ and $z_{ijq} = z_{ijr} = 0$ for all $j \in V, q \in \Q_j', r \in \R_j'$. We adjust $\eta_j$ and $\phi_j$ for each $j \in V$ to preserve constraint (\ref{eq:single3}). This new solution has $W_i = R_i = 0$ and has at least as large surplus and profit. Since we set all LP variables corresponding to $i$ to zero, this satisfies constraints (\ref{eq:open0}), (\ref{eq:flowbalance0}), and (\ref{eq:lowerbound0}), and is therefore feasible for the LP.

\medskip
We therefore only focus on facilities whose $W_i > 0$. Consider the set of these facilities and split them into two groups. Let 
$$S_1 = \{ i \in \F \ | y_i \in (0,1) \mbox{ and } R_i < 0\} \qquad \mbox{and} \qquad S_2 = \{ i \in \F \ | y_i \in (0,1) \mbox{ and } R_i \ge 0\}$$
Assume that for all of these facilities, there is no $j \in J_D(i)$, such that $j$ is {\em fully demand-utilized} and no $j \in J_S(i)$ such that $j$ is {\em fully supply-utilized}

First consider the facilities in set $S_2$, we can increase the LP variables till the condition of the lemma is satisfied; this process only increases both profit and surplus, preserving all constraints. We do this as follows: Suppose no $j  \in J_D(i)$  is fully demand-utilized and no $j \in J_S(i)$ is fully supply-utilized. In this case, let 
$$ \theta = \min\left( \frac{1}{y_i}, \min_{j \in J_D(i)} \left( \frac{1 -  \sum_{q \in \Q_j'} \sum_{i' \neq i}  z_{i'jq}}{\sum_{q \in \Q_j'}  z_{ijq} }\right), \min_{j \in J_S(i)} \left( \frac{1 -  \sum_{r \in \R_j'} \sum_{i' \neq i}  z_{i'jr}}{\sum_{r \in \R_j'}  z_{ijr} } \right) \right)$$  
Since $\eta_j < 1$ for all $j \in J_D(i)$ and $\phi_j < 1$ for all $j \in J_S(i)$, we have $\theta > 1$. Suppose we increase $y_i$, $z_{ijq}$ for all $j \in J_D(i), q \in \Q_j'$, and $z_{ijr}$ for all $j \in J_S(i)$, $r \in \R_j'$ by a factor of $\theta$. We will still have $\eta_j \le 1$ for all $j \in J_D(i)$ and $\phi_j \le 1$ for all $j \in J_S(i)$. However, either $y_i$ or one of these values will become exactly $1$. Note that since we scaled all LP variables corresponding to $i$ by the same factor, this preserves constraints (\ref{eq:open0}), (\ref{eq:flowbalance0}), and (\ref{eq:lowerbound0}). The surplus and profit of this facility increase by a factor $\theta > 1$, which contradicts the optimality of the LP solution. Therefore, the facilities in $S_2$ all have a neighboring $j$ that is either fully demand-utilized or fully supply-utilized.

Next consider the facilities in set $S_1$. Suppose the condition in the lemma is not satisfied, so that there are two facilities $i$ and $i'$ with $y_i, y_{i'}  \in (0,1)$, and with no neighboring $j$ that is either fully demand-utilized or fully supply-utilized. Suppose $W_i/|R_i| = a$ and $W_{i'}/|R_{i'}|= b$ with $a \ge b$. We multiply each LP variable corresponding to $i$ by a factor of $(1+\delta)$, and multiply each LP variable corresponding to $i'$ by a factor of $\left(1 - \frac{W_{i}}{W_{i'}}  \delta\right)$. Using the same argument as above, this process preserves the constraints that are specific to a facility, since all variables are changed by the same factor.  The increase in surplus of facility $i$ is $\delta W_i$, and the decrease in surplus of facility $i'$ is $\delta W_i$, so the overall surplus is preserved. The decrease in profit of facility $i$ is $|R_i| \delta$, and the increase in profit of facility $i'$ is $|R_{i'}| \frac{W_{i}}{W_{i'}}  \delta \ge |R_i| \delta$ by our assumption that $a \ge b$. Therefore, this process cannot decrease profit, hence all constraints are preserved. We choose $\delta$ as the smallest value that either makes facility $i$ have $y_i = 1$ or one neighboring $j$ either fully supply or demand utilized, or that sets the variables of facility $i'$ to zero. In all cases, the size of set $S_1$ reduces by one. We repeat this process till there is only one facility in $S_1$, completing the proof.
\end{proof}

The following corollary now restates Lemma~\ref{lem:struct0}, completing its proof.
 
\begin{corollary}
\label{cor1}
There is a $(1+\epsilon)$-approximation to the LP optimum where any facility with $y_i > 0$ satisfies the following condition:  either $y_i = 1$; or there exists $j \in J_D(i)$, such that $j$ is {\em fully demand-utilized}; or there exists $j \in J_S(i)$ such that $j$ is {\em fully supply-utilized}.
\end{corollary}
\begin{proof}
By Constraint~(\ref{eq:tighter2}), there is a set of facilities $S$ that are fully open ({\em i.e.}, $y_{i} = 1$) and $\sum_{i \in S} W_{i} \ge \W \frac{1-\epsilon}{\epsilon}$ by Constraint~(\ref{eq:tighter3}). The rounding in Lemma~\ref{lem:struct2} does not touch these facilities, since we only increase/decrease variables corresponding to partially open facilities ({\em i.e.}, those with $y_i \in (0,1)$). Lemma~\ref{lem:struct2} implies there is only facility $i$ that violates the condition of the corollary. This facility must have surplus $W_i \le \W y_i \le \W$ by Constraint~(\ref{eq:tighter1}).  This means closing this facility (setting all its associated variables to zero) reduces the LP optimum by at most a factor of $(1-\epsilon)$. Since the previous lemma implies this facility had $R_i  < 0$, this means closing it only increases profit, preserving weak budget balance.
\end{proof}

\section{Rounding the LP Relaxation}
\label{sec:rounding}
The rounding now follows approaches similar to those in~\cite{ZS16,williamson2011}. We first present the high-level idea. Note that if a node $j$ has $\alpha_{j0} = 0$ or $\beta_{j0} = 0$, then Constraint~(\ref{eq:open0}) implies $\sum_{i \in B_R(j)} y_i \ge 1$. Consider an {\em independent set} of such nodes, such that no two are fractionally assigned to the same facility. For any $j$ in this  set, move all partially open facilities in $B_R(j)$ to $j$ itself, so that there is a facility integrally opened at $j$. Since we move an entire facility, we preserve all flows, so that flow balance and lower bound are preserved, and so is profit. Now a demand/supply can be assigned a distance $2R$ away, and the opened facilities are integral.

At this point, consider any fractionally open facility $i$. It must have a node $j$ adjacent to it that satisfies the condition in Lemma~\ref{lem:struct0}. If $j$ has a facility completely open at its location, then move $i$ to location $j$. Otherwise, $j$ was not part of the independent set in the previous step, which means $j$ and $j'$ shared a fractionally open facility, and the previous step opened a facility completely at $j'$. In this case, we move $i$ to $j'$, again preserving all flows. This means any demand/supply moves distance at most $4R$, preserving all the LP constraints. 

\subsection{Rounding Facilities}
We now present the rounding algorithm in detail. Initially, all facilities  $i \in \F$ with $y_i > 0$ are {\em partially open}. Node $j \in V$ is {\em untouched} if for all $i$ such that $j \in J_D(i) \cup J_S(i)$, the facility $i$ is partially open.  Let $U$ be the set of {\em untouched} nodes, and let $Z$ be the set that is either fully demand-utilized or fully supply-utilized.  Let $U_f = U \cap Z$. 

\medskip
\noindent {\bf Phase 1.} Consider any $j \in U_f$. W.l.o.g., assume $\eta_j = 1$; the case where $\phi_j = 1$ is symmetric.  Let  $ N(j) = \{i | j \in J_D(i) \}$.  For every $i \in N(j)$, we ``move" $i$ to location $j$; call the new facility at location $j$ as $i^*$. This means we set
\begin{itemize}
\item $\bar{y}_{i^*} \leftarrow \sum_{i \in N(j)} y_i$ and $\bar{y}_i \leftarrow 0 \qquad \forall i \in N(j)$;
\item $\bar{z}_{i^*j'q} \leftarrow \sum_{i \in N(j)} z_{ij'q}$  and $\bar{z}_{ij'q} \leftarrow 0 \qquad \forall j' \in V, q \in \Q_{j'}, i \in N(j)$
\item $\bar{z}_{i^*j'r} \leftarrow \sum_{i \in N(j)} z_{ij'r}$  and $\bar{z}_{ij'r} \leftarrow 0 \qquad \forall j' \in V,  r \in \R_{j'}, i \in N(j)$
\end{itemize}

From constraint (\ref{eq:open0}), and the fact that $\eta_j = 1$, we have:  
$$\bar{y}_{i^*}  =  \sum_{i  \in N(j)}  y_i \ge \sum_{i \in B_R(j)} \sum_{q \in \Q_j'} z_{ijq} = \eta_j = 1$$ 
Subsequently, we mark every $i \in N(j)$  as {\em closed}, and mark $i^*$ as {\em completely open}. Furthermore, we mark every $j'$ that was reassigned in the above steps as {\em touched}.

Note that in the last three steps, any agents at a node $j'$ that was initially assigned to $i \in N(j)$ is now assigned to $i^*$. Since  each of the distances $j' \rightarrow i$ and $i \rightarrow j$ is at most $R$, the distance from $j'$ to $i^*$ is at most $2R$. Therefore, this step relaxes the distance of a feasible assignment to a facility from $R$ to $2R$.

This process trivially preserves the objective and weak budget balance, as well as constraints (\ref{eq:single3}). Moreover, constraints  (\ref{eq:open0}) and (\ref{eq:flowbalance0}) are satisfied since we add both sides of the constraints corresponding to $i \in N(j)$ to obtain the constraint for $i^*$. Finally, to see that (\ref{eq:lowerbound0}) is satisfied for $i^*$, note that 
$$ \sum_{j' \in J_D(i^*)}   \sum_{q \in \Q_{j'}'}  d_{j'} q \bar{z}_{i^*j'q}   =  \sum_{i \in N(j)} \sum_{j' \in J_D(i)} \sum_{q \in \Q_{j'}'} d_{j'}  q z_{ij'q}  \ge  \L \sum_{i \in N(j)} y_i  \ge   \L
$$
    
We continue this process, finding a node $j \in U_f$, and merging all facilities in $N(j)$ to one location. At the end of this process, the set $U_f$ is empty. 

\medskip
\noindent {\bf Phase 2.} At the end of Phase 1, each node $j$ which is {\em touched } (including all {\em fully utilized} nodes) route some fraction of their demand (or supply) to at least one  facility that is completely open. However, there could still be partially open facilities with $y_i \in (0,1)$ to which demand and supply are assigned. Consider these partially open facilities in arbitrary order. Suppose we are considering facility $i$ and there exists {\em touched} and {\em fully utilized} node $j$ such that $\sum_{q \in \Q_j'} z_{ijq} > 0$ (resp. $\sum_{r \in \R_j'} z_{ijr} > 0$).  Consider the completely open facility $i^*$ such that $\sum_{q \in \Q'_j} \bar{z}_{i^*jq} > 0$ or $\sum_{r \in \R'_j} \bar{z}_{i^*jr} > 0$. We move the facility $i$ to location $i^*$, updating the variables just as in Phase 1; {\em i.e.}, we set
\begin{itemize}
\item $\bar{y}_{i^*} \leftarrow \bar{y}_{i^*} +  y_i$ and $\bar{y}_i \leftarrow 0$;
\item $\bar{z}_{i^*j'q} \leftarrow \bar{z}_{i^*j'q} + z_{ij'q}$  and $\bar{z}_{ij'q} \leftarrow 0 \qquad \forall j' \in V, q \in \Q_{j'}$
\item $\bar{z}_{i^*j'r} \leftarrow \bar{z}_{i^*j'r} + z_{ij'r}$  and $\bar{z}_{ij'r} \leftarrow 0 \qquad \forall j' \in V,  r \in \R_{j'}$
\end{itemize}

The argument that all constraints are preserved follows just as before. For any $j'$ that was partially assigned to $i$, the new assignment is to $i^*$. This distance is at most 
$$c(j',i^*) \le c(j',i) + c(i,j) + c(j,i^*) \le R + R + 2R = 4R$$
 where we note that the distance $j \rightarrow i^*$ was at most $2R$ because $j$ was potentially reassigned to $i^*$ in Phase 1.  We mark all nodes $j'$ that are reassigned in this process as {\em touched}.

\medskip
At the end of this process, suppose there are still partially open facilities with $y_i \in (0,1)$. By Corollary~\ref{cor1}, each of these facilities $i$ must have some $j \in Z$ partially assigned to it. At the end of Phase 1, we have the invariant that $j \notin U_f$, since $U_f$ is empty. This means $j$ was touched on Phase 1. But in that case, $i$ must have been reassigned in Phase 2, which is a contradiction.  Therefore, at this point, all  facilities are either closed ($\bar{y}_i = 0$) or completely open ($\bar{y}_i \ge 1$). Furthermore, for any variable $\bar{z}_{ijq} > 0$ (resp. $\bar{z}_{ijr} > 0$), the facility $i$ is completely open; the distance from $j$ to $i$ is at most $4R$; for each completely open facility, the rate of supply equals the rate of demand (Constraint (\ref{eq:flowbalance0})), and finally, the total flow is at least $\L$ (Constraint (\ref{eq:lowerbound0})).  

\subsection {Final Steps and Running Time} 
At this point, the facilities are opened integrally. Lemma~\ref{lem:increase} now implies that we can choose one price/wage per node preserving all constraints and the objective.  This completes the proof of the following theorem.

\begin{theorem}
\label{thm:oblivious}
For any constant $\epsilon > 0$, there is a  $(4,1+\epsilon)$ approximation algorithm for \TSFL$(\L,R)$. 
\end{theorem}

Note that the surplus  can become arbitrarily close to zero. Therefore, for parameter $\Delta > 0$, we will allow additive error $\Delta$ in the surplus objective. Note that the maximum possible surplus is $W_{\max} = (\sum_j d_j) p_{\max}$, which is an upper bound on $\W$. If we assume the surplus is at least $\Delta$, then $\max_i W^*_i \ge \Delta/n$. Since the top $1/\epsilon$ facilities on $OPT$ have surplus $\W (1-\epsilon)/\epsilon$, this means we can set $\W \ge \frac{ \epsilon \Delta}{2n}$ Therefore,  the number of choices of $\W$ is $O\left( \frac{1}{\epsilon} \log \frac{n W_{\max}}{\epsilon \Delta} \right)$. For each choice of $\W$, we need to solve $O(n^{1/\epsilon})$ LPs, so that the overall number of LPs is $O\left( \frac{n^{1/\epsilon}}{\epsilon} \log \frac{n W_{\max}}{\epsilon \Delta} \right)$.

Note that $\Delta$ can be exponentially small, and our algorithm for solving the LP in Appendix~\ref{sec:compute} will lose such an additive factor in the objective anyway. Though we omit details, the same techniques work for approximating profit as well as approximating throughput (volume of trade) subject to weak budget balance. If the objective is profit, we can achieve optimal objective by directly rounding the {\em single LP} in Section~\ref{sec:basic} (details similar to Section~\ref{sec:elicited}). If the objective is throughput (or total flow) subject to weak budget balance, then the flow to any open facility is at least $\L$, so that we only need to guess the $\theta$ open facilities in OPT and not the flow to them. This means we  only need to solve $O(n^{1/\epsilon})$ LPs to achieve a $(1+\epsilon)$ approximation to throughput, again relaxing the distance constraint by a constant factor.

%% file: justif.tex
\section{Queueing-theoretic Justification: Dynamic Marketplaces}
\label{sec:justif}
In the facility location model discussed above, we imposed a lower bound $\L$ on the flow routed to any facility. In the absence of this constraint, {\em i.e.}, when $\L =0$, the problems are in fact poly-time solvable. For instance, to maximize surplus (gains from trade), we can simply write the standard welfare maximization LP ignoring prices and wages. Its dual yields prices and wages that will be strongly budget balanced. Therefore, the entire hardness of the problem comes from the lower bound constraint.  This begs the question: {\em Why have this constraint at all?}

 
 \medskip
 \noindent{\bf Dynamic Marketplace Model.}  We now present a  {\em dynamic marketplace} model that provides queueing-theoretic justification for these constraints. This model has the following features:
 
 \begin{itemize}
 \item Buyer and seller types are located in a metric space just as before. 
 \item Buyers and sellers are {\em no longer fluid}. Instead, buyers at node $j$ arrive as a {\tt Poisson} process with rate $d_j F_j(p)$ when quoted price $p$, and when quoted a price $p$; similarly, sellers follow a {\tt Poisson} process with rate $s_j H_j(w)$.
 \item  Each buyer and seller has a private patience level or deadline; if not matched within their deadline, they abandon the system. The platform knows the patience distribution.
 \end{itemize}
 
The stochastic control problem that we term {\em dynamic marketplace problem} can be summarized by two control decisions: 
\begin{enumerate}
\item {\em Pricing decision.} Choose static prices $p_j$ and wages $w_j$ at each node $j \in V$;  and 
\item {\em Scheduling decision.} This matches feasible buyer-seller pairs and removes them from the system.  This decision is dynamic, depending on the entire state of the system as captured by the number of unmatched buyers and sellers at different nodes at any point of time.  
\end{enumerate}
The goal is to design a stochastic control policy that maximizes the long-term average surplus subject to long-term budget balance. We insist all scheduled matches must involve a current buyer and seller with metric distance at most $R$. The key difference 
is in the {\em service availability guarantee}: Given the stochastic nature of our arrivals, 
there is always some probability that an incoming buyer or seller exhausts her patience before being matched. A more realistic goal is to design policies that guarantee a minimum level of service availability. We quantify this via the long-term average probability of abandonment of agents.  Formally, given a parameter $\eta > 0$ as input, the goal of the platform is to make the abandonment probability at most $\eta$.

\medskip
\noindent {\bf Scheduling Policies.} Constructing the optimal policy for the dynamic marketplace problem is closely related to several lines of work in {\em dynamic matchings} over a compatibility graph -- in kidney exchanges~\cite{AkbarpourGharan} where patients abandon the exchange if their health fails; in control of matching queues for housing allocation~\cite{caldentey2009fcfs}; and more generally in service system design~\cite{gurvich2014dynamic,adan2012exact,adan2015reversibility}, wherein customers and servers arrive stochastically and are matched according to a compatibility graph.  In all these models, the choice of whom to match an arriving agent to depends on the entire set of agents waiting at different nodes, leading to the ``curse of dimensionality". 


Given this curse of dimensionality, we  consider the restricted sub-class of matching policies where the platform creates  facilities in the metric space, and uses each facility to cater to a different set of mutually-compatible agent types. Arriving agents are randomly routed to a compatible facility, where they are queued up to be matched to agents on the other side. The probabilistic routing is fixed over time, and does not depend on the state of the  facilities.

Each  facility maintains a queue of active buyers and sellers that have been assigned there, ignores what location they came from, and matches them up using an optimal scheduling policy for minimizing abandonment rate using only the current state of that particular queue. We will enforce the constraint that for any  facility, the long-term abandonment probability is at most $\eta$, which in turn will ensure the overall abandonment probability is at most $\eta$. 

Suppose we assume buyer deadlines are distributed as {\tt Exponential}$(\kappa)$, and seller deadlines are distributed as {\tt Exponential}$(\gamma)$. Though these distributions are known, the scheduling decisions at any facility are made without knowing the patience level of any individual agent. Then, any work-conserving policy (including FIFO) is optimal. If an agent's deadline expires and there is no agent to match it with in the queue, this agent is considered abandoned.  We build on results from queueing theory~\cite{WardGlynn03} to bound this abandonment rate tightly as follows:

\begin{theorem}
\label{thm:queue}
Suppose $\lambda$ and $\mu$ be the (Poisson) arrival rates of buyers and sellers into a  facility. Assume buyer deadlines are distributed as {\tt Exponential}$(\kappa)$, and seller deadlines are distributed as {\tt Exponential}$(\gamma)$. Then the FIFO policy has abandonment rate at most $\eta$ when:
\begin{enumerate}
\item There is {\em flow balance}, that is, $\lambda = \mu$; and 
\item There is a {\em flow lower bound}, that is, $ \lambda \ge \frac{3}{2} \left(\frac{\min(\gamma, \kappa)}{\eta^2}\right)$.
\end{enumerate}
For $\eta \le \frac{1}{6}$, the above conditions are also necessary to a constant factor: For the abandonment probability to be at most $\eta$, it must hold that $\lambda/\mu \in [1-\eta, 1 + \eta]$; and $\min(\lambda, \mu) \ge \frac{1}{14000} \left(\frac{\min(\gamma, \kappa)}{\eta^2}\right)$.
\end{theorem}
\begin{proof}
The behavior of a  facility is captured via the following birth-death Markov chain: consider the state-space $\{\ldots,s(2),s(1),0,b(1),b(2),\ldots\}$, where $0=s(0)=b(0)$ denotes the state that the facility is empty, while for any $n\geq1$, the state $b(n)$ denotes that there are $n$ buyers queued up, and state $s(n)$ denote that there are $n$ sellers queued up.
For any $n\geq 1$, the transition rate from $b(n)$ to $b(n+1)$ is $\lambda$, and for $s(n)$ to $s(n+1)$ is $\mu$; on the other hand, the rate of transition from $s(n)$ to $s(n-1)$ is $n \gamma + \lambda$, while from $b(n)$ to $b(n-1)$ is $n \kappa + \mu$. 
Here, the term $n \gamma$ corresponds to the rate of abandonment of sellers as their deadlines expires, and  $n\kappa$ is the rate of abandonment of buyers.

We now show the sufficient conditions, and relegate the proof of the necessary conditions to Appendix~\ref{app:queueing}.  Assume that $\lambda=\mu$. 
Let $q_0$ denote the steady state probability of the queue being empty. Let 
$$ \Pr[\mbox{State } = s(n)] = \alpha_n \qquad \Pr[\mbox{State } = b(n)] = \beta_n$$
where $\alpha_0 = \beta_0 = q_0$. We have the following balance equations:
$$ \alpha_n n \gamma = \lambda ( \alpha_{n-1} - \alpha_n) \qquad \mbox{and} \qquad  \beta_n n \kappa = \lambda ( \beta_{n-1} - \beta_n)$$
Adding these equations, we have:
$ \sum_{n=1}^{\infty} \left( \alpha_n n \gamma +  \beta_n n \kappa \right) = 2 \lambda q_0$.
Note that the LHS here is the total abandonment rate, and since the total arrival rate of agents (buyers and sellers) is $2 \lambda$, this means the abandonment probability is exactly $q_0$. 

Since $ \alpha_n = \frac{\lambda}{\lambda + n \gamma} \alpha_{n-1}$ and $ \beta_n = \frac{\lambda}{\lambda + n \kappa} \beta_{n-1}$, we have by telescoping:
\begin{eqnarray*}
 \alpha_n & = q_0 \frac{\lambda^n}{\prod_{j=1}^{n} (\lambda + j \gamma)} = q_0 \prod_{j=1}^{n} \frac{1}{\left( 1 + j \frac{\gamma}{\lambda} \right)} \\
\beta_n &= q_0  \frac{\lambda^n}{\prod_{j=1}^{n} (\lambda + j \kappa)} = q_0 \prod_{j=1}^{n} \frac{1}{\left( 1 + j \frac{\kappa}{\lambda} \right)} 
\end{eqnarray*}
Since these probability values sum to one, this implies
\begin{equation}
\label{eq:q0} 
\frac{1}{q_0} = 1 +  \sum_{n=1}^{\infty} \left(  \prod_{j=1}^{n} \frac{1}{\left( 1 + j \frac{\gamma}{\lambda} \right)} + \prod_{j=1}^{n} \frac{1}{\left( 1 + j \frac{\kappa}{\lambda} \right)}  \right)
\end{equation}
For given $\kappa$ and $\gamma$, this is an increasing function of $\lambda$. Therefore, $q_0 \le \eta$ translates to a bound of the form  $ \lambda \ge \L$. An upper bound $\L_{\eta}$ on $\L$ can be computed as follows. Let $c = \frac{\min(\gamma,\kappa)}{\lambda}$. Then,
\begin{align*}
\frac{1}{q_0} &\ge  1+ \sum_{n=1}^{\infty} \left(  \prod_{j=1}^{n} \frac{1}{\left( 1 + j \frac{\min(\gamma,\kappa)}{\lambda} \right)} \right) \\
& \ge     \sum_{n=0}^{\infty} e^{- c n^2/2}  - e^{-c} \ge  \int_0^{\infty} e^{- c x^2/2}dx - e^{-c}  =  \sqrt{\frac{\pi}{2c}} - e^{-c} \ge \sqrt{\frac{2}{3c}}
\end{align*}
where the second inequality uses $1+x \le e^x$ for all $x \ge 0$. Therefore, if we insist $\sqrt{\frac{2\lambda}{3\min(\gamma,\kappa)}} \ge \frac{1}{\eta}$, this ensures the abandonment probability is at most $\eta$. This translates to the following lower bound:
$$ \lambda \ge \L =  \frac{3}{2}\cdot \frac{\min(\gamma, \kappa)}{\eta^2}$$
This completes the proof of the sufficient conditions. We present the proof of the necessary condition in Appendix~\ref{app:queueing}.
\end{proof}

We have therefore shown that a sufficient condition for bounding the abandonment probability at any  facility by $\eta$ reduces to saying the flow to the facility is balanced, and facility is {\em thick} -- there is a lower bound $\L = \frac{3}{2} \left(\frac{\min(\gamma, \kappa)}{\eta^2}\right)$ on how much demand or supply needs to be routed there. This reduces the {\em dynamic} control problem with atomic agents to a {\em static} problem where demand/supply are fluid -- exactly \TSFL$(\L,R)$ for suitable $\L$ that depends on $\eta$. Note that $\eta$ is an input parameter, so that $\L$ need not be a constant value.

%% file: Elicitation0.tex
\newcommand{\LL}{\mathcal{L}}
\newcommand{\w}{\mathcal{G}}

\section{Envy-Free Pricing and Profit Maximization}
\label{sec:elicitation}
\label{sec:elicited}
The idea of independently scaling up/down LP variables corresponding to individual facilities is fairly general, and leads naturally to approximation algorithms for more complex variants that are motivated by different scheduling policies for the dynamic marketplace problem. In this section, we present one such formulation that generalizes the model discussed in Section~\ref{sec:oblivious}. In section ~\ref{app:queueing2}, we show that this model corresponds to the setting when the platform uses prices to elicit patience of agents, and uses Earliest Deadline First (EDF) scheduling in each facility. 

We assume each node (type) $j$ of buyer/seller has a collection of subtypes $\S_j$. There is a DAG $G_j(\S_j,E_j)$ on $\S_j$ that captures {\em envy}. If there is an edge $(k_1,k_2) \in E_j$, then sub-type $k_1$ envies sub-type $k_2$. The platform announces a price (resp. wage) $p_{jk}$ (resp. $w_{jk}$) for each sub-type $k \in \S_j$. In order to preserve incentive compatibility, we require that if $(k_1,k_2) \in E_j$, then $p_{jk_1} \le p_{jk_2}$; resp. $w_{jk_1} \ge w_{jk_2}$. This prevents an agent of sub-type $k_1$ from reporting its type to be $k_2$. Note that since the graph $G_j$ is a DAG, such a price (resp. wage) assignment is feasible. We term such an assignment of prices (resp. wages) at each $j$ as a {\em price (resp. wage) ladder}. 

As before, there is a non-increasing demand function $d_{jk} F_{jk}(p_{jk})$ for each buyer sub-type $k \in \S_j$, and a  non-decreasing supply function $s_{jk} H_{jk}(w_{jk})$ for each seller sub-type $k \in \S_j$. Each sub-type $k \in \S_j$ is also associated with a weight $\w_{jk}$. The platform learns which sub-type any agent chooses. 

\medskip \noindent {\bf  Lottery Pricing and Assignment.} 
The platform opens a set of facilities. For each node $j$ and $k \in \S_j$, buyers (resp. sellers) arriving at the node and choosing that type are probabilistically routed to facilities which are within distance $R$ from the node. We assume the platform shows a lottery over price (resp. wage) ladders as follows: For each node $j \in V$ the platform maintains a distribution $\Z_j$ of  facilities within distance bound $R$, and for each facility in this set, it maintains a distribution $\LL_{ij}$ of price (resp. wage) ladders. Given an agent arriving at this node, the platform first chooses a facility $i$ from $\Z_j$, and then a ladder from $\LL_{ij}$ and shows it to the agent. After the agent chooses the price or wage (hence revealing its sub-type), she is routed to facility $i$. We note that the routing policy makes the facility the agent is routed to be independent of the sub-type elicited. Though this assumption is somewhat restrictive, it prevents the agent from choosing a sub-type to optimize for the facility they get assigned to.

\medskip
\noindent {\bf Service Availability Guarantee.} As before, we capture service availability by ensuring that each facility $i$ has balanced supply and demand, and is also sufficiently {\em thick}. However, we now capture thickness by a lower bound $\L$ on the {\em total weight} of the sub-types assigned there. Formally, let $x_{ijk}$ denote the expected flow of sub-type $k \in \S_j$ to facility $i$.

\begin{description}
\item[Flow Balance.] The expected amount of supply and demand assigned there are the same.
\item[Weight Lower Bound.] The expected weight assigned there is large:  $\sum_{j \in V, k \in \S_j} \w_{jk} x_{ijk} \ge \L$.
\end{description} 

\newcommand{\EFFL}{{\sc Envy-Free FL}}

The objective is to maximize the expected profit of the solution.  We term this problem \EFFL$(\L,R)$.  We note that similar ideas can be used to maximize other objectives; we present the profit objective for simplicity. In section ~\ref{app:elicitation}, we prove the following theorem:

\begin{theorem}
\label{thm:elicited}
There is a polynomial time $(4,1)$ approximation for \EFFL$(\L,R)$. That is, we obtain the optimal expected profit by relaxing the distance constraint by a factor of $4$. 
\end{theorem}

In the dynamic marketplace setting presented in section ~\ref{app:queueing2}, the sub-types correspond to different deadlines, and the weight of a sub-type is precisely the deadline value. We show there that the weight lower bound corresponds to the condition for the EDF scheduling policy to have low abandonment rate.

%% file: ElicitationModel.tex
\subsection{Approximation Algorithm for \EFFL$(\L,R)$}
\label{app:elicitation}
Our LP formulation and rounding are similar to the one for \TSFL$(\L,R)$, and we highlight the differences. As before, we assume there is a candidate set $\mathcal P$ and $\mathcal W$ of prices and wages for each node, respectively. The set of all candidate  facilities in the metric space is denoted by $\F$; since we assume the metric space is explicitly specified as input, we set $\F = V$.  For each node $j$, $B_R(j) \subseteq \F$ denotes the set of all the  facilities $i \in \F$ such that $c(i,j) \le R$. For each  facility $i$, define $B_R(i)$ as the set of all the nodes $j \in V$ such that $c(i,j) \le R$. 

\subsubsection{Linear Programming Relaxation}
For each candidate  facility $i \in \F$, let $y_i \in \{0,1\}$ be the indicator variable that a  facility exists at that location in the metric space. These are the only integer variables in our formulation. Variables $x^d_{ij}$ and $x^s_{ij}$ are non-zero only if $y_i = 1$ and $i \in B_R(j)$. In this case, those are respectively the probability that buyers and sellers at node $j$ are routed to  facility $i$. Note that there is some probability that all prices at node $j$ are set to $p_{\max}$, which corresponds to not routing node $j$ anywhere. Let $z_{ijkp}$ be the probability that buyers at node $j$ with sub-type $k \in \S_j$ are assigned to  facility $i$ and offered price $p$. Similarly, $z_{ijkw}$ denotes the probability that sellers at node $j$ with sub-type $k$ are assigned to  facility $i$ and offered wage $w$. 


\medskip
\noindent{\bf Objective and Constraints.} The objective is the same as before.

\begin{equation}
\label{eq:rev} \max\ \ \sum_{j \in V} \sum_{k\in \S_j}  \sum_ {i \in B_R(j)} \left( \sum_{p \in \P} p d_{jk} F_{jk}(p) z_{ijkp}  - \sum _{w \in \W} w s_{jk}H_{jk}(w) z_{ijkw}  \right)
\end{equation}

The following constraints connect the variables together. We present these constraints only for buyers (that is, $p \in \P$); the constraints for sellers is obtained by replacing $p$ and $x^d$ with $w  \in \W$ and $x^s$. Since we route the buyers at node $j$ probabilistically to one of the  facilities, or to no facility by offering all deadlines a price $p_{\max}$:
\begin{equation}
\label{eq:route}
\sum_{i \in B_R(j)} x^d_{ij}  \le  1 \quad \forall j \in V
\end{equation}
Next, a price should be offered to each buyer with sub-type $k$ at node $j$ assigned to facility $i$: 
\begin{equation}
\label{eq:priceEdge}
\sum_{p \in \P} z_{ijkp}  =  x^d_{ij}\quad \forall j\in V, i\in B_R(j), k \in \S_j
\end{equation}  
Next, if demand is fractionally routed from $j$ to $i$, then $i$ should be open and within distance $R$:
\begin{equation}
\label{eq:open}
x^d_{ij}  \le  y_i\quad \forall j \in V, i\in B_R(j)
\end{equation} 

We next enforce that the prices and wages form a distribution over ladders. Note that the policy first chooses the  facility to route to, and then chooses from a distribution over ladders. This reduces to a stochastic dominance condition for the distributions corresponding to $z$:

\begin{eqnarray}
\label{eq:priceladder}
\sum_{p'\le p, p' \in \P} z_{ijkp'} & \le & \sum_{p' \le p, p' \in \P} z_{ijk'p'} \quad \forall p \in \mathcal P, (k,k') \in E_j, \forall j \in V, \forall i \in B_R(j)\\ 
\sum_{w'\le w, w' \in \W} z_{ijkw'} & \ge & \sum_{w' \le w, w' \in \W} z_{ijk'w'} \quad \forall w \in \mathcal W, (k,k') \in E_j, \forall j \in V, \forall i \in B_R(j)
\end{eqnarray}	

Finally, we encode the service availability constraints.  We first capture {\em flow balance} at each  facility: the rate of arrival of buyers and sellers are equal.
\begin{equation}
\label{eq:flowbalance}
\sum_{j \in B_R(i)} \sum_{k \in \S_j, p\in \P} d_{jk} F_{jk}(p) z_{ijkp} =  \sum_{j \in B_R(i)} \sum_{k \in \S_j, w\in \W} s_{jk} H_{jk}(w) z_{ijkw}  \quad \forall i \in \F
\end{equation}
We finally encode {\em weighted flow lower bound} on the total deadline of buyers and sellers at the facility: 

\begin{equation}
\label{eq:lowerbound}
\sum_{j \in B_R(i)} \sum_{k \in \S_j} \w_{jk} \left( \sum_{p\in \P} d_{jk} F_{jk}(p) z_{ijkp} + \sum_{w\in \mathcal W} s_{jk} H_{jk}(w)  z_{ijkw} \right) \ge \L y_i \quad \forall i \in \F
\end{equation}

\subsubsection{Rounding}
If we ignore the integrality constraints on $y_i$, the above is a linear programming relaxation of the problem.  We will now show how to round the resulting solution.

We generalize Lemma~\ref{lem:struct} using the following definitions of {\em fully utilized}. We say that $j \in V$ is {\em fully demand utilized} if $\sum_{i \in B_R(j)} x^d_{ij} = 1$; similarly, it is {\em fully supply-utilized} if $\sum_{i \in B_R(j)} x^s_{ij} = 1$. We say $j$ is {\em partially demand-connected} to facility $i \in \F$ if $x^d_{ij} > 0$, and {\em partially supply-connected} if $x^s_{ij} > 0$. Let $J_D(i)$ denote the set of nodes that are partially demand-connected to $i \in \F$, and $J_S(i)$ be the set that is partially supply-connected.    As before, we define the profit of a  facility  $i \in \F$ as:
$$R_i = \sum_{j \in B_R(i)} \sum_{k\in \S_j}  \left( \sum_{p \in \P} p d_{jk} F_{jk}(p) z_{ijkp}  - \sum _{w \in \W} w s_{jk}H_{jk}(w) z_{ijkw}  \right)$$

\begin{lemma}
In the LP optimum, for any $i \in \F$, $R_i \ge 0$. Further,  if $R_i > 0$, either $y_i = 1$; or there exists $j \in J_D(i)$, s.t. $j$ is {\em fully demand-utilized}; or there exists $j \in J_S(i)$, s.t. $j$ is {\em fully supply-utilized}.
\end{lemma}

The proof of the above lemma follows the same argument as Lemma~\ref{lem:struct}: If a facility has negative $R_i$, we can set all its variables to zero without violating any constraints. If the condition in the Lemma is violated for $i \in \F$, then we can increase all variables corresponding to $i$ by the same factor till the condition is satisfied. Since all constraints involve single facilities, this process preserves them while increasing the objective. For this transformation to work, it is crucial Constraints (\ref{eq:priceladder}) are defined separately for each $(i,j)$ pair; in other words, we crucially need to assume the policy chooses a facility first and then chooses a distribution over ladders for that facility.

\medskip 
The rounding now proceeds in the same way as in Section~\ref{sec:oblivious}: In Phase 1, we identify untouched and {\em fully utilized} $j$ and merge all $i$ to which it is partially connected to one facility. Note that the total $y_i$ of these facilities is at least $1$ by the LP constraints. At the end of this phase, we move the remaining partially open $i$ as in Phase 2 of the rounding scheme. This preserves the profit, and satisfies the flow balance and lower bound constraints ($B_R$ is replaced by $B_{4R}$ in the constraints), yielding the following theorem:

\begin{theorem}
\label{thm2}
There is a feasible solution $\{\bar{x}, \bar{y}, \bar{z}\}$ to the above linear program, whose objective is optimal, and all of whose constraints are satisfied. For each $i \in \F$, either $\bar{y}_i = 0$ or $\bar{y}_i \ge 1$.
\end{theorem}

\subsubsection{Final Policy} The final choice of prices and wages, and the routing policy is the following. We present it only for buyers; the policy for sellers is symmetric.

\begin{itemize}
\item At node $j$, choose a facility $i$ with probability $\bar{x}_{ij}$. If no facility is chosen, the price is set to $p_{\max}$. 
\item If facility $i$ is chosen, choose $\alpha$ uniformly at random in $[0,1]$. For each $k \in \S_j$, find $p_k \in \P$ such that $\sum_{p' < p_k, p' \in \P} \frac{\bar{z}_{jkp'}}{\bar{x}_{ij}} \le \alpha$ and $\sum_{p' \le p_k, p' \in \P} \frac{\bar{z}_{jkp'}}{\bar{x}_{ij}} > \alpha$. Post prices $\{p_1,p_2, \ldots, p_K\}$.
\item If the buyer accepts price $p_k$, route her to  facility $i$.
\end{itemize}

Constraints (\ref{eq:priceladder}) imply that regardless of the choice of $\alpha$, the prices $\{p_1,p_2, \ldots, p_K\}$ in the second step form a ladder, so that $p_1 \ge p_2 \ge \cdots \ge p_K$. A similar statement holds for wages. Therefore, the second step produces a lottery over ladders. Further, if $Z_{ijkp}$ denote the event that the price for sub-type $k \in \S_j$ is $p$ and facility $i$ is chosen, then it is an easy exercise to check that $\E[Z_{ijkp}] = \bar{z}_{jkp}$. Therefore, the randomized policy exactly implements the solution found in Theorem~\ref{thm2}, so that it maximizes profit.  Omitting details, this completes the proof of Theorem~\ref{thm:elicited}.

%% file: queueing.tex
\renewcommand{\L}{\mathbf L}

\subsection{Justification of \EFFL\  via Dynamic Marketplace Model}
\label{app:queueing2}
In this section, we present a dynamic marketplace model that justifies the problem statement of \EFFL. As in Section~\ref{sec:justif},  we assume buyers and sellers have an inherent patience level or deadline. If they are not matched within their deadline, they drop out of the system.  We assume every agent $m$  is associated with a patience level $\nu_m$; unlike Section~\ref{sec:justif}, we do not assume these are Exponentially distributed.  The platform advertises a fixed set of patience levels, or deadlines, denoted by $\S_j = \{\nu_{j1}, \nu_{j2}, \ldots, \nu_{jK}\}$, which is a guarantee on the time by which a buyer or seller choosing that deadline is guaranteed to be matched. We assume $\nu_{j1} \le \nu_{j2} \le \cdots \le \nu_{jK}$.  For simplicity, we use $k \in \S_j$ and $\nu_k \in \S_j$ interchangeably.

\medskip \noindent {\bf  Incentive-compatibility.} 
We assume the platform sets a lottery of  prices and wages at each node $j$, that are independent of time. Consider the issue of eliciting deadlines truthfully. Consider buyers first. At node $j$, suppose the platform offers price $p_{jk}$ for deadline $\nu_{jk}$. Every buyer can choose one deadline in $\S_j$, in which case he pays price $p_{jk}$, and is guaranteed to be matched within time $\nu_{jk}$ from his arrival. We assume any buyer  $m$ has very large negative utility for being matched after his patience level $\nu_m$, therefore he will choose a $k$ such that $\nu_{jk} \le \nu_m$. Subject to this, he will choose $k$ with smallest $p_{jk}$, since this maximizes his valuation minus price.   A symmetric model can be posited for sellers, where we replace price with wage, and the seller chooses the largest wage such that the corresponding deadline  is smaller than his own patience level.

Since the goal of the platform is to elicit patience levels truthfully,  the platform chooses a {\em price ladder} $p_{j1} \ge p_{j2} \ge \cdots \ge p_{jK}$ and {\em wage ladder} $w_{j1} \le w_{j2} \le \cdots \le w_{jK}$ at each node $j$. This ensures that agents with $\nu_m \in [\nu_{jk},\nu_{jk+1}]$ report deadline $\nu_{jk}$.

\medskip
Each deadline level $\nu_{jk} \in \S_j$ gets associated with non-increasing demand function $d_{jk} F_{jk}(p_{jk})$, which is the Poisson rate at which buyers $m$ with patience $\nu_m \in [\nu_{jk},\nu_{jk+1}]$ arrive when the price of deadline $\nu_{jk}$ is $p_{jk}$. Similarly, deadline level $k \in \S_j$  is associated with a non-decreasing supply function $s_{jk} H_{jk}(w_{jk})$, which is the Poisson rate at which sellers $m$ with patience $\nu_m \in [\nu_{jk},\nu_{jk+1}]$ arrive when the wage for deadline $\nu_{jk}$ is $w_{jk}$. These deadline levels correspond to the sub-types described before.

\medskip \noindent {\bf Scheduling Policy.}   As in Section~\ref{sec:justif}, the platform opens a set of facilities. For each node $j$ and deadline level $k$, buyers (resp. sellers) arriving at the node and choosing that deadline are probabilistically routed to  facilities which are within distance $R$ from the node. Buyers and sellers arriving at the  facility are queued up, and optimally matched to minimize abandonment. Since the platform knows which deadline was chosen by the agent,  the optimal matching policy is now a variant of Earliest Deadline First (EDF): When the deadline of some buyer (resp. seller) expires, it is matched to that seller (resp. buyer) in the queue whose deadline will expire earliest in the future. If an agent's deadline expires and there is no agent to match it with in the queue, this agent is abandoned. It is an easy exercise to show that this policy maximizes the number of matches made in any  facility.  

As in Section~\ref{sec:justif}, the goal of the platform is to design a joint pricing and scheduling policy to maximize profit, while ensuring bounded match distance and bounded abandonment probability.

\subsubsection{Bounding Abandonment Rate} 
We will now show that the weight lower bound can be interpreted as a sufficient condition for the abandonment rate of the EDF policy to be at most $\eta$, where the weight of a sub-type is simply its deadline value. 

The main technical assumption we require in this part is that the desired abandonment probability, $\eta$ is small, in particular that $\eta \ll \frac{\nu_{\min}}{\nu_{\max}}$.  As noted above, the scheduling policy within a  facility is a variant of EDF. Unlike the {\sc Patience-oblivious}  model where the behavior of a  facility could be modeled as a variant of a $M/M/1$ queue, the optimal abandonment probability in a two-sided EDF queue clearly depends on the entire distribution of deadlines of buyers and sellers, which in turn depends on the pricing scheme and assignment policy. However, we crucially need a closed-form bound on this probability in order to plug into an LP relaxation for the overall problem. We use recent results from queueing due to Kruk {\em et al.}~\cite{KrukLRS03} to construct such a closed-form bound, whose very existence we find non-trivial and surprising!  

Kruk {\em et al.}~\cite{KrukLRS03} present an approximation to the abandonment probability of a one-sided queue $M/M/1$ queue with EDF scheduling. They approximate the queueing process via a reflected Brownian motion. We adapt their result to our setting, and rephrase it below. Consider the queue associated with a  facility. Let $\bar{S}$ denote the average deadline of a seller arriving to this queue, and $\bar{D}$ denote the average deadline of a buyer arriving to the queue.  Note that the distribution of deadlines and the arrival rate depends on the overall pricing and assignment policies. 

\medskip
Recall that we assumed $\eta$ is small, in particular that $\eta \ll \frac{\nu_{\min}}{\nu_{\max}}$. We first enforce that supply and demand arrive to the queue at the same rate; call this rate $\lambda$.   Next suppose w.l.o.g. that $\bar{D} > \bar{S}$. Consider the policy that instantaneously matches arriving sellers to the queued buyer with earliest deadline; if the queue is empty, the seller is abandoned. This exactly mimics a one-sided $M/M/1$ queue with EDF scheduling. We quote the following result informally from~\cite{KrukLRS03}: 

\begin{quote}
Consider a one sided $M/M/1$ queue with arrival rate and service rate equal to $\lambda$. Suppose deadlines of jobs are independently distributed with mean $\bar{D}$, and the scheduler uses the EDF policy. Then holding $\lambda$ and $\frac{\nu_{\max}}{\nu_{\min}}$ fixed, in the regime where $\nu_{\min}$ becomes very large, the abandonment probability approaches $\frac{1}{\lambda \bar{D}}$. 
\end{quote}

Though part of their argument is heuristic, they perform simulations to show that this approximation is indeed accurate. Since we need abandonment probability of $\frac{1}{\lambda \bar{D}}$ to be at most  $\eta \ll 1$, and since we assumed $\frac{\nu_{\min}}{\nu_{\max}} \gg \eta$, this automatically enforces that all deadlines are much larger than the mean inter-arrival time, satisfying their precondition for our setting. 

Since the optimal policy for a two-sided queue only has lower abandonment probability, we use $\frac{1}{\lambda \bar{D}}$ as an upper bound on this quantity. Since we assumed $\bar{D} \ge \bar{S}$, we will instead use  $\frac{2}{\lambda (\bar{D} + \bar{S})}$ as the upper bound, which we will set to be at most $\eta$. 

\medskip
We now show that this is the best possible upper bound that only depends on $\bar{D}$ and $\bar{S}$. Suppose buyers deadlines are deterministic with value $\bar{D}$, and seller deadlines are deterministic with value $\bar{S}$. Then the optimal policy matches without waiting in a FIFO fashion. This means the loss probability assuming the queue has buyers is the same as that of a $M/M/1$ queue with deadlines $\bar{D}$, which from~\cite{Boots} is exactly
$$ P_1 = \frac{1}{1+\lambda \bar{D}}$$
Similarly, when there are sellers in the queue, the loss probability is
$$ P_2 = \frac{1}{1+\lambda \bar{S}}$$
Conditioned on the queue being empty and a buyer arriving, the expected time after which the queue next becomes empty is $T_b = \frac{1}{P_1} = 1 + \lambda \bar{D}$, in which period the loss probability is $P_1$. Similarly, if a seller arrives when the queue is empty, the expected time after which the queue again becomes empty is $T_s = \frac{1}{P_2} = 1 + \lambda \bar{D}$, in which period the loss probability is $P_2$. Since a buyer or seller arrives with equal probability when a queue is empty, the expected loss probability is
$$ P = \frac{T_b P_1 + T_s P_2}{T_b + T_s} = \frac{2}{2 + \lambda \bar{D} + \lambda \bar{S}} \approx \frac{2}{\lambda  \left(\bar{D} + \bar{S} \right)}$$
assuming $\lambda  (\bar{D} + \bar{S}) \gg 1$.

\medskip
In summary, each  facility needs to satisfy the following two sufficient conditions for its abandonment probability to be at most $\eta$:

\begin{enumerate}
\item The rate of arrival of supply and demand should be the same; call this rate $\lambda$.
\item If $\bar{S}$ denote the average deadline of a seller, and $\bar{D}$ denote the average deadline of a buyer, then $\lambda (\bar{D} + \bar{S}) \ge \frac{2}{\eta}$.
\end{enumerate}

Therefore, to reduce the scheduling policy to an instance of \EFFL$(R,\L)$, we set $\w_{jk} = \nu_{jk}$ and $\L = \frac{2}{\eta}$, so that the second condition above translates to the weight lower bound. This justifies the \EFFL$(R,\L)$ problem as capturing the optimal scheduling  policy for the dynamic marketplace problem presented above.

Note that the resulting lower bound on $\lambda$ derived by the above condition is a significant improvement over the patience-oblivious case, since the lower bound now depends on $\frac{1}{\eta}$ instead of $\frac{1}{\eta^2}$. This intuitively means that in order to achieve comparable profit and abandonment probability, we can aim for a higher quality of match by reducing the radius $R$.  A similar observation that even partial information about deadlines significantly reduces abandonment is made in~\cite{AkbarpourGharan}, albeit for a different model.

%% file: appendix.tex
\section{Omitted Proofs and Examples}

\input{example.tex}

\subsection{Solving the LP Formulation in Section~\ref{sec:oblivious}}
\label{sec:compute}
We now show how to use the Ellipsoid algorithm to efficiently solve the LP formulation in Section~\ref{sec:oblivious} to arbitrary additive accuracy even when the demand and supply distributions are continuous, so that the sets $\Q_j$ (resp. $\R_j$) are continuous. First we get rid of weak budget balance by take a Lagrangian of surplus and the profit. For any parameter $\lambda \ge 0$, define:
$$ \V^{\lambda}_j(q) = \V_{j}(q) + \lambda d_j q F^{-1}_{j}(q)$$
and 
$$ \C^{\lambda}_j(r) = \C_{j}(r)  + \lambda  s_j r H^{-1}_{j}(r) $$

Since we assumed regular supply and demand distributions, it is easy to show that $\V^{\lambda}_j(q)$ is concave in $q$ and $\C^{\lambda}_j(r)$ is convex in $r$. The Lagrangian objective is then:
$$ \mbox{Maximize}  \sum_{j \in V}  \left( \sum_{q \in \Q_j} \sum_{i \in B_R(j)} z_{ijq} \V^{\lambda}_{j}(q)   - \sum _{r \in \R_j} \sum_{i \in B_R(j)} z_{ijr} \C^{\lambda}_{j}(r)  \right)$$
\[
\begin{array}{rcll}
\sum_{q \in \Q_j} \sum_{i \in \F}  z_{ijq}   & \le &  1 \quad & \forall j \in V \\
\sum_{r \in \R_j} \sum_{i \in \F}  z_{ijr}   & \le &  1 \quad & \forall j \in V \\
\sum_{q \in \Q_j} z_{ijq}  & \le  & y_i  & \forall j \in V, i\in B_R(j) \\
\sum_{q \in \R_j} z_{ijr}  & \le  & y_i  & \forall j \in V, i\in B_R(j) \\
\sum_{j \in B_R(i)} d_j  \sum_{q \in \Q_j} q z_{ijq} & = &  \sum_{j \in B_R(i)} s_j \sum_{ r \in \R_j}  r z_{ijr}  & \forall i \in \F \\
\sum_{j \in B_R(i)} d_j   \sum_{q \in \Q_j}  q z_{ijq}  & \ge & \L y_i & \forall i \in \F \\
z_{ijq}, z_{ijr}, y_i & \ge & 0 & \forall i,j,q,r
\end{array}
\]

The dual is the following:

$$ \mbox{Minimize} \sum_{j \in V} (a_j + b_j)$$
\[
\begin{array}{rcll}
a_j + \eta_{ij} + d_j q ( \zeta_i - \rho_i) & \ge & \V^{\lambda}_j(q) & \forall  j \in V, i \in B_R(j), q \in \Q_j \\
b_j + \theta_{ij} - s_j r \zeta_i  + \C^{\lambda}_j(r) & \ge & 0 &  \forall  j \in V, i \in B_R(j), r \in \R_j \\
\L \rho_i & \ge & \eta_{ij} + \theta_{ij} & \forall j \in V, i \in B_R(j) \\
\eta_{ij}, \theta_{ij},  \rho_i & \ge & 0 & \forall j \in V, i \in B_R(j) 
\end{array}
\]

For fixed dual variables, since $\V^{\lambda}_j(q)$ is concave in $q$ and $\C^{\lambda}_j(r)$ is convex in $r$, it is easy to check that for each $i,j$, the separation oracle either involves maximizing a concave function  in $q$ (for the first set of constraints) or minimizing a convex function  in $r$ (for the second set of constraints). In either case, finding the separating hyperplane involves one-dimensional convex optimization. This implies the LP admits to an efficient additive approximation even for continuous distributions over a bounded domain. We omit the standard details.

\subsection{Proof of Theorem~\ref{thm:queue}: Necessary Condition}
\label{app:queueing}
For the necessary conditions, the first condition is obvious: If $\lambda/\mu \notin [1-\eta, 1+\eta]$, in steady state, a fraction $\eta$ of either buyers or sellers must necessarily be discarded just because an equal number of buyers and sellers are matched. 

Next, we show that if $\lambda = \mu$, the bound on $q_0$ in Equation (\ref{eq:q0}) is tight to a constant factor. Recall we define 
$c = \min(\gamma,\kappa)/\lambda$. First, note that if $c\geq 1$, we have
\begin{align*}
\frac{1}{q_0} &\le 2 \sum_{n=0}^{\infty} \left(  \prod_{j=0}^{n} \frac{1}{\left( 1 + jc\right)} \right) 
\le 2 \sum_{n=0}^{\infty} \left(  \frac{1}{(n+1)!} \right)  = 2\cdot \left(e - 1\right)   
\end{align*}
This gives $q_0\geq 0.29$ which is infeasible as we want $\eta\leq 1/6$. Thus, the only relevant case is $c = \min(\kappa,\gamma)/\lambda\leq 1$, for which we have:
\begin{align*}
\frac{1}{q_0} &\le 2 \sum_{n=0}^{\infty} \left(  \prod_{j=0}^{n} \frac{1}{\left( 1 + jc\right)} \right) \\ 
&\le 2  \left( \sum_{n=0}^{1/\sqrt{c}-1} 1 + \sum_{n=1/\sqrt{c}}^{\infty} \left(\prod_{j=1}^{1/\sqrt{c}} \frac{1}{\left( 1 + jc \right)}\right)\left( \frac{1}{1+\sqrt{c}}  \right)^{n-1/\sqrt{c}} \right) \\
&\le 2  \left( \frac{1}{\sqrt{c}}  + 2^{-\frac{c}{2}\left(\frac{1}{\sqrt{c}}\right)^2} \sum_{i=0}^{\infty} \left( \frac{1}{1+\sqrt{c}}  \right)^i \right) \qquad\mbox{(Since $1/(1+x)\leq 2^{-x}$ for $x\in [0,1]$)}\\
&\le 2  \left( \frac{1}{\sqrt{c}}  + \frac{1}{\sqrt{2}}\cdot \frac{1+\sqrt{c}}{\sqrt{c}}   \right)
\le   \frac{7}{\sqrt{c}}
\end{align*}

Next suppose $\lambda > \mu$; note that this means $\lambda \in \mu \cdot [1, 1+\eta]$ since $\lambda/\mu \in [1-\eta,1+\eta]$. Now consider the subset of the Markov chain on the states $\{0,b(1),b(2),\ldots\}$, with transition rates from $b(n)$ to $b(n-1)$ is $\mu + n \kappa$, and that from $b(n-1)$ to $b(n)$ is $\lambda > \mu$; we henceforth refer to this as the {\em buyer system}. Note that the abandonment rate of buyers in this system is stochastically dominated by the abandonment rate if $\mu = \lambda$. 


Now, conditioned on the queue being empty and a buyer (resp. seller) arriving, let $T_b$ (resp. $T_s$) denote the expected time after which the queue next becomes empty, and let $R_b$ (resp. $R_s$) denote the abandonment rates in these periods. When the queue is empty, the probability that a buyer arrives is $\frac{\lambda}{\lambda + \mu} \in \left[\frac{1}{2} , \frac{7}{13} \right]$, since $\eta \le 1/6$. Moreover, Wald's identity gives that the overall abandonment rate is by 
$ R\left(\frac{\lambda}{\lambda+\mu}T_b+ \frac{\mu}{\lambda+\mu}T_s\right) = \left(\frac{\lambda}{\lambda+\mu}R_b T_b+ \frac{\mu}{\lambda+\mu}R_s T_s\right)$, and thus
$$ R \ge \frac{\mu}{\lambda} \left(\frac{T_b R_b + T_s R_s}{T_b + T_s}\right) \ge \frac{6}{7}  \times R_b \times \frac{T_b}{T_b + T_s}$$

Now we need to consider two separate cases:
\begin{enumerate}
\item $\gamma>\kappa$: Here, the transition rate from $s(n)$ to $s(n-1)$ is $\lambda + n \gamma$, and that transition rate from $b(n)$ to $b(n-1)$ is $\mu + n \kappa < \lambda + n \gamma$; on the other hand, the rate from $s(n-1)$ to $s(n)$ is $\mu$ and that from $b(n-1)$ to $b(n)$ is $\lambda > \mu$. 
By stochastic dominance, we therefore have $T_s \le T_b$, which means the abandonment rate is:
$$ R \ge  \frac{6}{7}  \times R_b \times \frac{T_b}{T_b + T_s} \ge \frac{3}{7} R_b$$ 
Moreover the abandonment probability in the buyer system is smaller if we assume $\lambda = \mu$, and hence $\frac{R_b}{\lambda} \ge  \frac{1}{7} \sqrt{\frac{\kappa}{\lambda}}$. Combining these bounds ,we have:
$$ R \ge \frac{3}{7} R_b \ge \frac{3}{49} \sqrt{\kappa \lambda}$$
Since the overall arrival rate is at most $2 \lambda$, the abandonment probability is at least $\frac{1}{33} \sqrt{\frac{\kappa}{\lambda}}$.  
\item $\gamma\le\kappa$: In this case, note that $T_b$ is the inverse of the probability the buyer system is in state $b(0)$. Since this probability is maximized when $\lambda = \mu$, we have 
$T_b \ge  \left(\sqrt{\frac{3 \mu}{2\kappa}}\right)$.
Similarly, $T_s$ is maximized when $\lambda = \mu$, so that  $T_s \le 7 \sqrt{\frac{\lambda}{\gamma}}$.
Therefore, assuming $\eta \le 1/6$, we have:
$$\frac{T_s}{T_b} \le \left(\frac{7 \sqrt{2}}{\sqrt{3}}\sqrt{\frac{\kappa}{\gamma}}\right) \le 6 \sqrt{\frac{\kappa}{\gamma}}\Rightarrow \frac{T_b}{T_s+T_b}\geq \frac{1}{1+6 \sqrt{\kappa/\gamma}}\geq \frac{1}{7}\sqrt{\frac{\gamma}{\kappa}}$$
Further, as before, the abandonment rate $R_b \ge  \frac{1}{7}\sqrt{\lambda \kappa}$, and thus 
$$R \ge \frac{6}{7} \times \frac{1}{7} \sqrt{\lambda\kappa} \times \frac{1}{7}\sqrt{\frac{\gamma}{\kappa}} = \frac{6}{343}\sqrt{\lambda \gamma}$$
This means the abandonment probability is at least $\frac{3}{343}\sqrt{\frac{\gamma}{\lambda}}$. 
\end{enumerate}
Combining the two, we see that to guarantee that in order to 
ensure average abandonment rate is at most $\eta$, we need $\lambda\geq \frac{1}{14000}\min(\gamma,\kappa)/\eta^2$.

%% file: example.tex
\subsection{Integrality Gap Example in Section~\ref{sec:welfare}}
\label{app:example}
First, we show that the optimum solution can open a  facility with negative profit. To be more specific, for any given constant $c<1$ we present a simple example in which $c$ fraction of the total surplus is generated by a facility with negative profit. Then we use this example to show that the LP has unbounded integrality gap. 

Let $V=\F=\{ v,v'\}$ such that $c(v,v') = \infty$, and $L$ be the lower bound for the total amount of demand (supply) at each open facility. For node $v$, assume $d_v=s_v=L$, the valuation of buyers is uniformly distributed over the interval $[2,3]$, and the cost of sellers is uniformly distributed over the interval $[ 0,1 ]$. For node $v'$, assume $d_{v'}=s_{v'}=L$, the valuation of buyers is uniformly distributed over the interval $[ c'-1,2c'+1 ]$, and the cost of sellers is uniformly distributed over the interval $[ 0,c' ]$ where $c'=\frac{2c}{1-c}$. We claim that the optimum integral solution for this example is to open a  facility at each of the nodes and set the price and wage at node $v$ to $2$ and $1$ respectively, and set the price and wage at node $v'$ to $c'-1$ and $c'$ respectively. 

First, we show that this solution is feasible. At each node the price is not more than the valuation of any arriving buyer. Therefore, all the buyers choose to participate. Similarly, since the wage is not less than the cost of any arriving seller, all the sellers choose to participate. This solution satisfies flow balance for each of the facilities because the volume of sellers and buyers are equal at the corresponding node, and all of them choose to participate. In addition, flow lower bound is also satisfied. Finally, the profit of the facility at $v$ is $d_v$ and the loss of the facility at node $v'$ is $d_{v'}$. Therefore, the total profit is 0 and  profit of the facility at node $v$ compensates for the loss at the other facility.

Now, we show that the surplus of the facility with negative profit is a fraction $c$  of the total surplus. The surplus at node $v$ is $d_v \times (2.5-0.5) = 2L$ and the surplus at node $v'$ is $d_{v'} \times (3c'/2 - c'/2)= c'L$. Therefore, $c' /(c'+2) = c$ fraction of the surplus is generated at node $v'$. 
     
Finally, we need to show that this solution is optimum. The nodes are far from each other and we cannot send the buyers and sellers from different nodes to a common facility. The only option for opening a facility at each of the nodes is to set the price and wage at each node in a way that all the arriving buyers and sellers choose to participate (otherwise, the flow lower bound cannot be satisfied). Therefore, this problem has three feasible integral solutions: no facility is opened, a facility at node $v$ is opened, and a facility at each of the nodes is opened. Note that the solution which only opens a facility at $v'$ is not feasible because it does not satisfy weak budget balance. The surplus of those solutions are 0, $2L$, and $(2+c')L$ respectively. Therefore, the third solution is optimum. Also it is easy to see that this integer optimum solution is also LP optimum solution. The reason is that both facilities generate positive surplus and partially opening any facility by fractionally assigning that node as an outlier results in lower surplus. 

\medskip
\noindent {\bf Integrality Gap.} Now we slightly modify the previous example to show that the LP has unbounded integrality gap. We only change the distribution of the valuation of the buyers at node $v'$.  The valuation of the buyers is now uniformly distributed over the interval $[ c'-1-\epsilon,2c'+1+\epsilon ]$ for a small positive constant $\epsilon$. After this change, the integral solution which opens a facility at each node is not feasible anymore because it violates weak budget balance constraint. Therefore, the optimum integral solution has $2L$ surplus. 

On the other hand we claim that there is a fractional solution which has $(\frac{1}{1+\epsilon} \times  \frac{2c}{1-c} + 2)L$ surplus. Set the price and wage at node $v$ to $2$ and $1$ and open the facility at that node ($y_v=1$). For the node $v'$ we can only open the facility partially. Set $y_{v'}=\frac{1}{1+\epsilon}$ and the price and wage at node $v'$ to $c'-1-\epsilon$ and $c'$ with probability $\frac{1}{1+\epsilon}$ and to $p_{max}$ and 0 with probability $\frac{\epsilon}{1+\epsilon}$. In other words, set $\alpha_{v'1} = \beta_{v'1}=\frac{1}{1+\epsilon}$ and $\alpha_{v'0}=\beta_{v'0}=\frac{\epsilon}{1+\epsilon}$. This solution is feasible and generates $(\frac{1}{1+\epsilon} \times  \frac{2c}{1-c} + 2)L$ surplus, while the optimum integer solution generates only $2L$ surplus. Note that $c$ can be arbitrarily close to 1 and therefore the integrality gap is unbounded.